\documentclass[runningheads]{llncs}
\usepackage[T1]{fontenc}
\usepackage{graphicx}
\usepackage{amsmath,amssymb,amsfonts}
\usepackage{algorithmic}
\usepackage{textcomp}
\usepackage{float}
\usepackage{caption}
\usepackage{subcaption}
\usepackage{wrapfig}

\usepackage[dvipsnames,svgnames]{xcolor}
\usepackage{enumitem}
\usepackage{comment}
\usepackage{listings}
\begin{document}
\title{Parallel-batched Interpolation Search Tree}

%
%
%
\author{Vitaly Aksenov\inst{1} \and
Ilya Kokorin\inst{1} \and
Alena Martsenyuk}
%
%
\institute{ITMO University, Russia }
%
\maketitle              
\begin{abstract}
A sorted set (or map) is one of the most used data types in computer science. In addition to standard set operations, like \texttt{Insert}, \texttt{Remove}, and \texttt{Contains}, it can provide set-set operations such as \texttt{Union}, \texttt{Intersection}, and \texttt{Difference}. Each of these set-set operations is equivalent to some batched operation: the data structure should be able to execute \texttt{Insert}, \texttt{Remove}, and \texttt{Contains} on a batch of keys. It is obvious that we want these ``large'' operations to be parallelized. These sets are usually implemented with the trees of logarithmic height, such as 2-3 trees, treaps, AVL trees, red-black trees, etc. Until now, little attention was devoted to data structures that work asymptotically better under several restrictions on the stored data. In this work, we parallelize Interpolation Search Tree which is expected to serve requests from a \emph{smooth} distribution in doubly-logarithmic time. Our data structure of size $n$ performs a batch of $m$ operations in $O(m \log\log n)$ work and poly-log span.

\keywords{Parallel Programming  \and Data Structures \and Parallel-Batched Data Structures.}
\end{abstract}

\section{Introduction}

A \emph{Sorted set} is one of the most ubiquitous \emph{Abstract Data Types} in Computer Science, supporting \texttt{Insert}, \texttt{Remove}, and \texttt{Contains} operations among many others. The sorted set can be implemented using different data structures: to name a few, skip-lists~\cite{pugh1990skip}, red-black trees~\cite{guibas1978dichromatic}, splay trees~\cite{sleator1985self}, or B-trees~\cite{graefe2011modern,comer1979ubiquitous}.

Since nowadays most of the processors have multiple cores, we are interested in parallelizing these data structures. There are two ways to do that: write a concurrent version of a data structure or allow one to execute a batch of operations in parallel. The first approach is typically very hard to implement correctly and efficiently due to different problems with synchronization. Thus, in this work we are interested in the second approach: \emph{parallel-batched data structures}.

Several parallel-batched data structures implementing a sorted set are presented: e.g., 2-3 trees~\cite{paul1983parallel}, red-black trees~\cite{park2001parallel}, treaps~\cite{blelloch1998fast}, (a, b) trees~\cite{akhremtsev2016fast}, AVL-trees~\cite{medidi1998parallel}, and generic joinable binary search trees~\cite{blelloch2016just,sun2018pam}.

Although many parallel-batched trees were presented, we definitely lack implementations that can execute separate queries in $o(\log n)$ time (where $n$ is the number of keys in the set) under some reasonable assumptions. However, there exist several sequential data structures with this property. One example is \emph{Interpolation Search Tree}, or \emph{IST}.

Despite the fact that concurrent IST is already presented~\cite{brown2020non,prokopec2020analysis} we still lack parallel-batched version of the IST: indeed, it differs from the concurrent version because concurrent IST allows many processes to execute scalar requests simultaneously, while parallel-batched IST utilizes multiprocessing to parallelize large non-scalar requests.
In this work, we design and test \emph{Parallel-batched Interpolation Search Tree}.

The work is structured as follows: in Section~\ref{preliminaries-section} we describe the important preliminaries;
in Section~\ref{IST-section} we present the original Interpolation Search Tree; in Sections~\ref{parallel-contains-section},~\ref{parallel-insert}, and~\ref{parallel-remove} we present the parallel-batched contains, insert and remove algorithms; in Section~\ref{parallel-rebuilding-section} we present a parallelizable method to keep the IST balanced; in Section~\ref{theoretical-section} we present a theoretical analysis; in Section~\ref{experiments-section} we discuss the implementation and present experimental results; we conclude in Section~\ref{conclusion-section}. 

\section{Preliminaries}
\label{preliminaries-section}
\subsection{Sorted set}

Sorted set stores a set of keys of the same comparable type \texttt{K} and allows one to execute, among others, the following operations:

\begin{itemize}
    \item \texttt{Set.Contains(key)}~--- returns \texttt{true} if \texttt{key} belongs to the set,  \texttt{false} otherwise.
    
    \item \texttt{Set.Insert(key)}~--- if \texttt{key} does not exist in the set, adds \texttt{key} to the set and returns \texttt{true}. Otherwise, it leaves the set unmodified and return \texttt{false}.
    
    \item \texttt{Set.Remove(key)}~--- if \texttt{key} exists in the set, removes \texttt{key} from the set and returns \texttt{true}. Otherwise, it leaves the set unmodified and return \texttt{false}.
    
    
\end{itemize}

\subsection{Parallel-batched data structures}
\label{parallel-batched-section}
\begin{definition}
\normalfont Consider a data structure $D$ storing a set of keys and an operation $Op$. If $Op$ involves only one key (e.g., it checks whether a single key exists in the set, or inserts a single key into the set) it is called a \emph{scalar operation}. Otherwise (i.e., if $Op$ involves multiple keys) it is called a \emph{batched operation}.

A data structure $D$ that supports at least one \emph{batched operation} is called a \emph{batched data structure}.
\end{definition}

We want a sorted set to implement the following batched operations:
\begin{itemize}
    \item \texttt{Set.ContainsBatched(keys[])}~--- the operation takes an array of keys of size $m$ and returns an array \texttt{Result} of the same size. For each $i \in [0; m)$, \texttt{Result[i]} is true if \texttt{keys[i]} exist in the set, and false otherwise.
    
    \item \texttt{Set.InsertBatched(keys[])}~--- the operation takes an array of keys of size $m$. For each $i \in [0; m)$, if \texttt{keys[i]} does not exist in the set, the operation adds it to the set.
    
    \item \texttt{Set.RemoveBatched(keys[])}~--- the operation takes an array of keys of size $m$. For each $i \in [0; m)$, if \texttt{keys[i]} exists in the set, the operation removes it from the set.
\end{itemize}

Note, that: 

\begin{itemize}
    \item \texttt{InsertBatched} calculates the union of two sets. Indeed, $\texttt{A} \leftarrow \texttt{A} \cup \texttt{B}$ can be written as \texttt{A.InsertBatched(B)};

    \item \texttt{RemoveBatched} calculates the difference of two sets. Indeed, $\texttt{A} \leftarrow \texttt{A} \setminus \texttt{B}$ can be written as \texttt{A.RemoveBatched(B)};

    \item \texttt{ContainsBatched} calculates the intersection of two sets. Indeed, $\texttt{A} \cap \texttt{B}$ can be calculated by identifying all keys, belonging in both \texttt{A} and \texttt{B}. As follows from the \texttt{ContainsBatched} definition, \texttt{A.ContainsBatched(B)} allows us to identify all such keys;
\end{itemize}

We can employ parallel programming techniques (e.g., \emph{fork-join parallelism}~\cite{blumofe1996cilk,lea2000java}) to execute batched operations faster. 

\begin{definition}
\normalfont A batched data structure $D$ that uses parallel programming to speed up batched operations is called a \emph{parallel-batched data structure}.
\end{definition}

\subsection{Time complexity model}
\label{work-span-section}

In our work, we assume the standard \emph{work-span} complexity model~\cite{acar2019algorithms} for \emph{fork-join} computations. We model each computation as a \emph{directed acyclic graph}, where nodes represent operations and edges represent dependencies between them: if there exists a path $v \rightarrow a_1 \rightarrow a_2 \rightarrow \ldots \rightarrow a_n \rightarrow u$ from operation \texttt{v} to operation \texttt{u}, then operation \texttt{v} must be executed before operation \texttt{u}. If neither path $v \mathrel{\leadsto} u$ nor path $u \mathrel{\leadsto} v$ exists, operations \texttt{u} and \texttt{v} can be executed in any order, even in parallel (see Fig~\ref{dag-pic}).  Moreover, the \emph{directed acyclic graph}, modeling the parallel computation, has exactly one \emph{source node} (i.e., the node with zero incoming edges) and exactly one \emph{sink node} (i.e., the node with zero outcoming edges). The only \emph{source node} corresponds to the start of the execution, while the only \emph{sink node} corresponds to the end of the execution.
Some nodes have two \emph{outcoming edges}~--- such operations spawn two parallel tasks and are called \emph{fork operations}.
Some nodes have two \emph{incoming edges}~--- such operations wait for two corresponding parallel tasks to complete and are called \emph{join operations}.

\begin{figure}[H]
  \centering
  \caption{Example of an \emph{directed acyclic graph}, representing a parallel execution\\} 
  \label{dag-pic}
  \includegraphics[width=\linewidth]{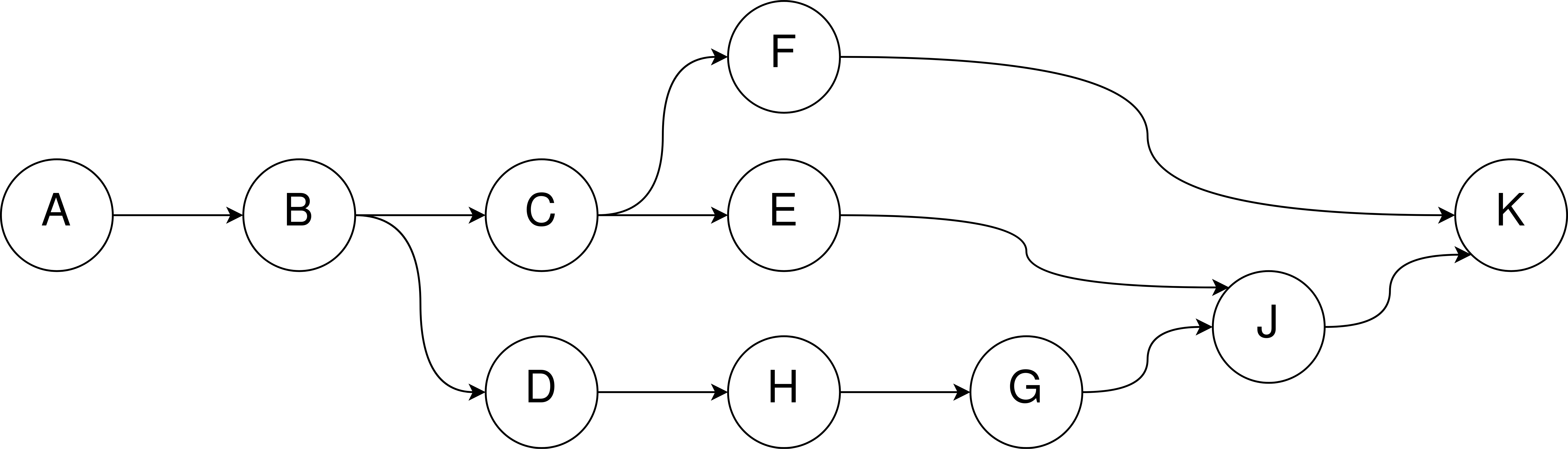}
\end{figure}

For example, in Fig.~\ref{dag-pic} operation \texttt{A} must be executed before operation \texttt{H} and operation \texttt{C} must be executed before operation \texttt{J}, while operations \texttt{E} and \texttt{G} can be executed in any order, even in parallel. Node \texttt{A} is the \emph{source node}, denoting the beginning of the execution, while node \texttt{K} is the \emph{sink node}, denoting the end of the execution. Nodes \texttt{B} and \texttt{C} correspond to \emph{fork operations}, while nodes \texttt{J} and \texttt{K} correspond to \emph{join operations}

Considering the execution graph of the algorithm, our target complexities are: 
\begin{itemize}
    \item \emph{work} denotes the number of nodes in the graph, i.e., the total number of executed operations. The graph in Fig.~\ref{dag-pic} has \emph{work} of value $10$.

    \item \emph{span} denotes the number of nodes on the longest path from the \emph{source node} to the \emph{sink node}, i.e., the length of the critical path in the graph. The graph in Fig.~\ref{dag-pic} has \emph{span} of value $7$, since the longest path from the \emph{source node} (node \texttt{A}) to the \emph{sink node} (node \texttt{K}) is the path $A \rightarrow B \rightarrow D \rightarrow H \rightarrow G \rightarrow J \rightarrow K$.
\end{itemize}

\subsection{Standard parallel primitives}
\label{primitives-section}

In this work, we use several standard parallel primitives. Now, we give their descriptions. Their implementations are provided, for example, in~\cite{jaja1992introduction}.

\textbf{Parallel loop. } It executes a loop body for \texttt{n} index values (from \texttt{left} to \texttt{right}, exclusive) in parallel. This operation costs $O(n)$ work and $O(\log n)$ span given that the loop body has time complexity $O(1)$. 

\begin{lstlisting}[escapeinside={(*}{*)},captionpos=t,numbers=none]
pfor i in left..right:
   loop_body(i)
\end{lstlisting}

\textbf{Scan. } \texttt{Result := Scan(Arr)} calculates \emph{exclusive} prefix sums of array \texttt{Arr} such that  \begin{equation*}
            Result[i] = 
             \begin{cases}
               0 & i = 0\\
               \sum\limits_{j = 0}^{i - 1} Arr[j] & i > 0
             \end{cases}
            \end{equation*}
\texttt{Scan} has $O(n)$ work and $O(\log n)$ span given that the sum of two values can be calculated in $O(1)$.
    
\textbf{Filter. } \texttt{Filter(Arr, predicate)} returns an array, consisting of elements of the given array \texttt{Arr} satisfying \texttt{predicate} keeping the order (e.g., \texttt{Filter([1, 3, 8, 6, 7, 2], is\_even)} returns \texttt{[8, 6, 2]}). \texttt{Filter} has $O(n)$ work and $O(\log n)$ span given that \texttt{predicate} has time complexity $O(1)$.
    
\textbf{Merge. } \texttt{Merge(A, B)} merges two sorted arrays \texttt{A} and \texttt{B} keeping the result sorted. \texttt{Merge} has $O(|A| + |B|)$ work and $O(\log^2 (|A| + |B|))$ span.

\textbf{Difference.} \texttt{Difference(A, B)} takes two sorted arrays \texttt{A} and \texttt{B} and returns all elements from \texttt{A} that are not present in \texttt{B}, in sorted order: e.g., \texttt{Difference([2, 4, 5, 7, 9], [2, 5, 9]) = [4, 7]}. \texttt{Difference} takes $O(|A| + |B|)$ work and $O(\log^2 (|A| + |B|))$ span.
    
\textbf{Rank. } Given that \texttt{A} is a sorted array and \texttt{x} is a value, we denote \texttt{ElemRank(A, x)} = $\vert \{e \in A \vert e \leq x\} \vert$ as the number of elements in \texttt{A} that are less than or equal to \texttt{x} (note, that it is also the insertion position of \texttt{x} in \texttt{A} so that \texttt{A} remains sorted after the insertion). For example, \texttt{ElemRank([1, 3, 5, 7], 2) = 1}, \texttt{ElemRank([1, 3, 5, 7], 5) = 3}, \texttt{ElemRank([1, 3, 5, 7], -1) = 0}.
Given that \texttt{A} and \texttt{B} are sorted arrays, we denote \texttt{Rank(A, B) = [$\texttt{r}_{\texttt{0}}$, $\texttt{r}_{\texttt{1}}$, $\ldots$, $\texttt{r}_{\vert \texttt{B} \vert - 1}$]}, where \texttt{$\texttt{r}_{\texttt{i}}$ = ElemRank(A, B[i])}. \texttt{Rank} operation can be computed in $O(|A| + |B|)$ work and $O(\log^2 (|A| + |B|))$ span.

\section{Interpolation Search Tree}
\label{IST-section}

\subsection{Interpolation Search Tree Definition}
\label{IST-structure-section}

Interpolation Search Tree (IST) is a multiway internal search tree proposed in~\cite{mehlhorn1993dynamic}.
IST for a set of keys $x_0 < x_1 < \ldots < x_{n - 1}$ can be either \emph{leaf} or \emph{non-leaf}. 

\begin{definition}
Leaf IST with a set of keys $x_0 < x_1 < \ldots < x_{n-1}$ consists of array \texttt{Rep} with $Rep[i] = x_i$, i.e., it keeps all the keys in a sorted array \texttt{Rep}.
\end{definition}

\begin{definition}
\emph{Non-leaf} IST with a set of keys $x_0 < x_1 < \ldots < x_{n - 1}$ consists of (Fig.~\ref{IST-array-pic} and~\ref{IST-from-array-pic}):
\begin{itemize}
    \item An array \texttt{Rep} storing an ordered subset of keys $x_{i_0}, x_{i_1}, \ldots x_{i_{k - 1}}$ (i.e., $0 \leq i_0 < i_1 < \ldots < i_{k - 1} < n$).
    
    \item Child ISTs $C_0, C_1 \ldots C_{k}$ (each can be either leaf or non-leaf), such that:

    \begin{itemize}
        \item $C_0$ is an IST storing a subset of keys $x_0, x_1 \ldots x_{i_0 - 1}$;
        \item for $1 \leq j \leq k - 1$, $C_j$ is an IST storing a subset of keys $x_{i_{j - 1} + 1}, \ldots x_{i_j - 1}$;
        \item $C_{k}$ is an IST storing a subset of keys $x_{i_{k - 1} + 1}, \ldots x_{n - 1}$;
    \end{itemize}
\end{itemize}
\end{definition}

\begin{figure}[H]
  \caption{Example of a non-leaf IST. Here $n = 13, k = 3, i_0 = 3, i_1 = 6, i_2 = 10$. Thus, $Rep[0] = x_3, Rep[1] = x_6, Rep[2] = x_{10}$. $C_0$ stores keys $x_0 \ldots x_2$, $C_2$ stores keys $x_4 \ldots x_5$, $C_3$ stores keys $x_7 \ldots x_{9}$, $C_4$ stores keys $x_{11} \ldots x_{12}$.}
  \label{IST-array-pic}
  \includegraphics[width=\linewidth]{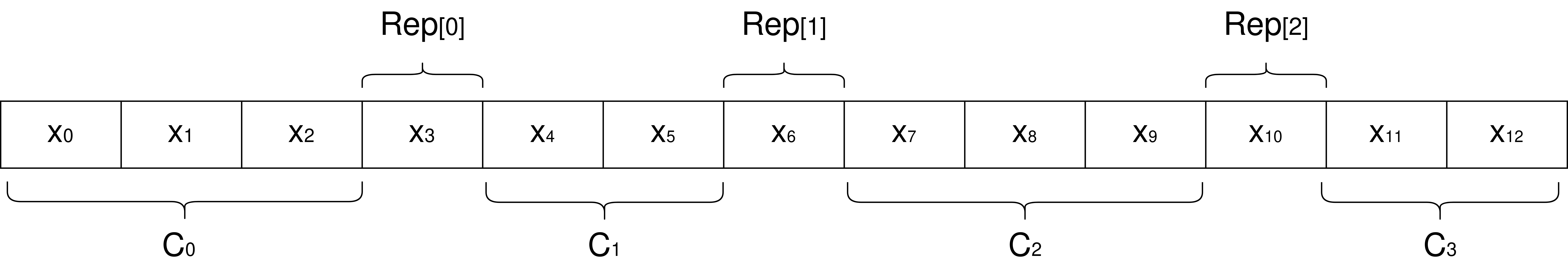}
\end{figure}

\begin{figure}[H]
  \centering
  \caption{Example of an IST built on array in Fig.~\ref{IST-array-pic}. Two nodes (root and its leftmost child) are non-leaf ISTs, while other nodes are leaf ISTs.}
  \label{IST-from-array-pic}
  \includegraphics[width=\linewidth]{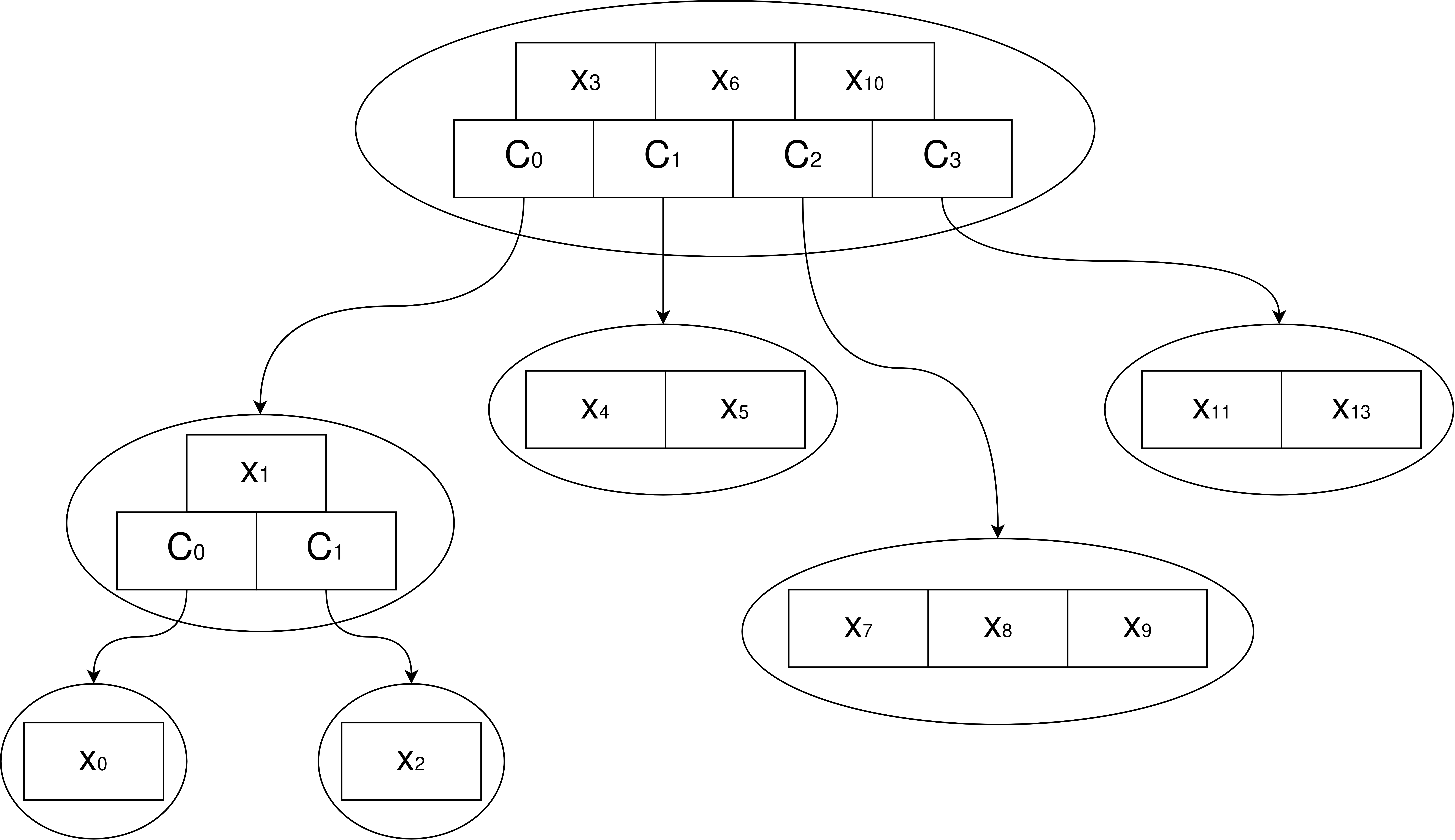}
\end{figure}

Any \emph{non-leaf} IST has the following properties: 

\begin{itemize}
    \item all keys less than $Rep[0]$ are located in $C_0$;
    \item all keys in between $Rep[j-1]$ and $Rep[j]$ are located in $C_j$;
    \item all keys greater than $Rep[k - 1]$ are located in $C_{k}$;
\end{itemize}

\subsection{Interpolation search and the lightweight index}
\label{interpolation-search-section}

We can optimize operations on ISTs with numeric keys, by leveraging the \emph{interpolation search technique}~\cite{peterson1957addressing,mehlhorn1993dynamic,willard1985searching}.
Each node of an IST (both leaf and non-leaf ones) has an \emph{approximate index} that can point to some place in the \texttt{Rep} array \emph{close} to the position of the key being searched. The structure of a non-leaf IST with an index is shown in Fig.~\ref{IST-pic}.

\begin{figure}[H]
  \centering
  \caption{Non-leaf IST contains: (1) \texttt{Rep} array with stored keys; (2) an array of pointers to child ISTs \texttt{C}; (3) an index, allowing for fast lookups of keys in the \texttt{Rep} array.} 
  \label{IST-pic}
  \includegraphics[width=\linewidth]{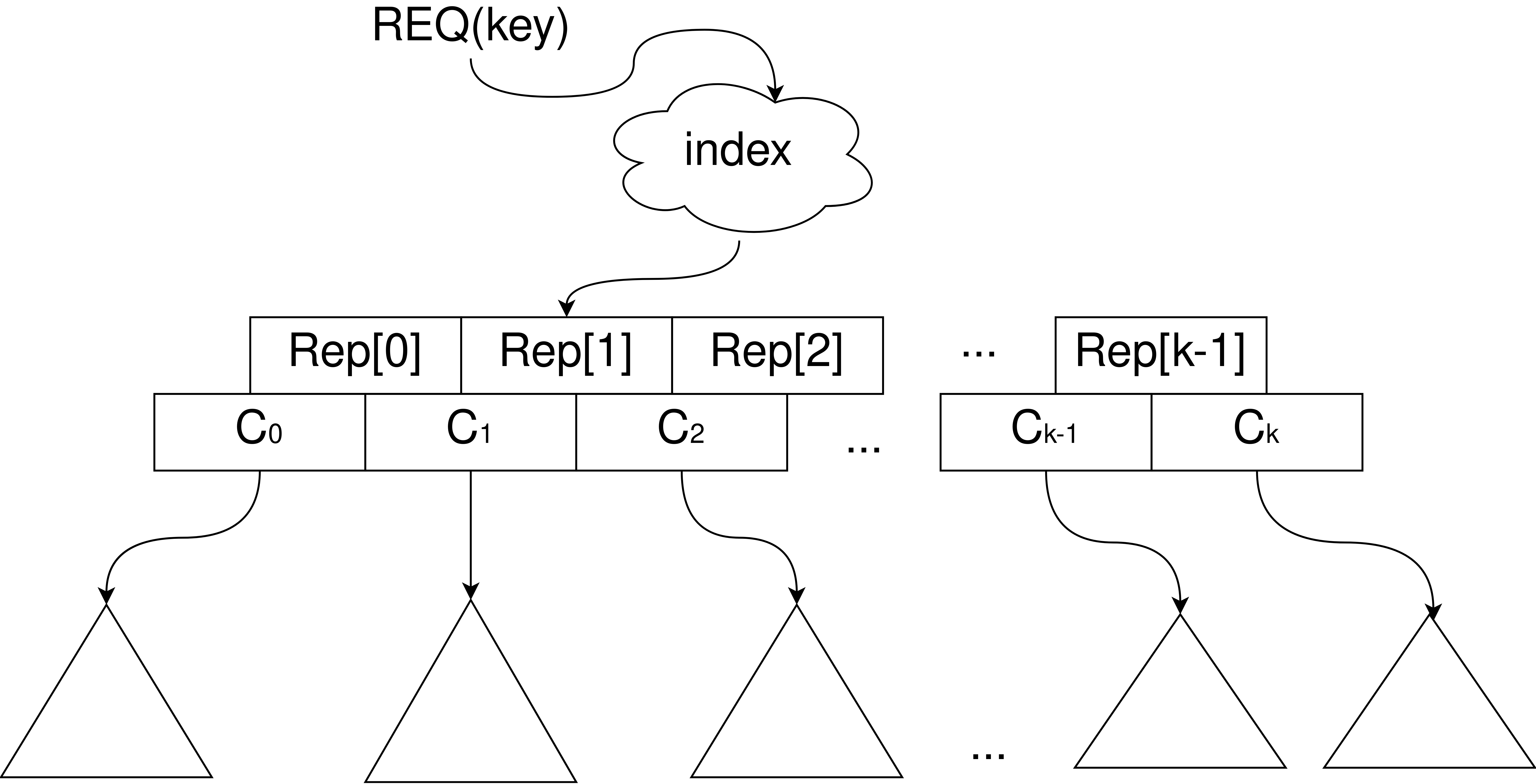}
\end{figure}

The original IST paper~\cite{mehlhorn1993dynamic} proposes to use an array \texttt{ID} of size \texttt{$m \in \Theta(n^\varepsilon)$} with some $\varepsilon \in [\frac{1}{2}; 1)$ as an index. \texttt{ID} array is defined the following way: let \texttt{a} and \texttt{b} be reals such that \texttt{a < b}. In that case,
\texttt{ID[i] = j} iff $Rep[j] < a + i \cdot \frac{b - a}{m} \leq Rep[j + 1]$. As stated in~\cite{mehlhorn1993dynamic}, \texttt{ID[$\lfloor \frac{x - a}{b - a} \cdot m \rfloor$]} is the approximate position of \texttt{x $\in (a; b)$} in \texttt{Rep}.

After finding the approximate location of \texttt{x} in \texttt{Rep}, we can find its exact location by using the linear search, as described in~\cite{mehlhorn1993dynamic}. Let us denote \texttt{i := ID[$\lfloor \frac{x - a}{b - a} \cdot m \rfloor$]}.

\begin{itemize}
    \item if \texttt{Rep[i] = x}, we conclude that we have found the exact position of \texttt{x};

    \item if \texttt{Rep[i] < x}, we conclude that \texttt{x} can only be found to the right of the \texttt{i}-th element in \texttt{Rep} array (since the \texttt{Rep} array is sorted). Therefore, we start incrementing \texttt{i} until we either: (1) find \texttt{x}; (2) find the element that is strictly greater than \texttt{x} or reach the end of \texttt{Rep} array~--- in that case we conclude that \texttt{x} is not contained in \texttt{Rep} array (Fig.~\ref{search-right-pic});

    \item if \texttt{Rep[i] > x}, we conclude that \texttt{x} can only be found to the left of the \texttt{i}-th element in \texttt{Rep} array (since the \texttt{Rep} array is sorted). Therefore, we start decrementing \texttt{i} until we either: (1) find \texttt{x}; (2) find the element that is strictly less than \texttt{x} or reach the beginning of \texttt{Rep} array ~--- in that case we conclude that \texttt{x} is not contained in \texttt{Rep} array (Fig.~\ref{search-left-pic}).
\end{itemize}

\begin{figure}[H]
    \caption{Determining the exact location of the key given the approximate location}
     \centering
    \begin{subfigure}[b]{0.45\linewidth}
          \centering
          \caption{Searching for the key on the right to the approximate position.}
          \includegraphics[width=\linewidth]{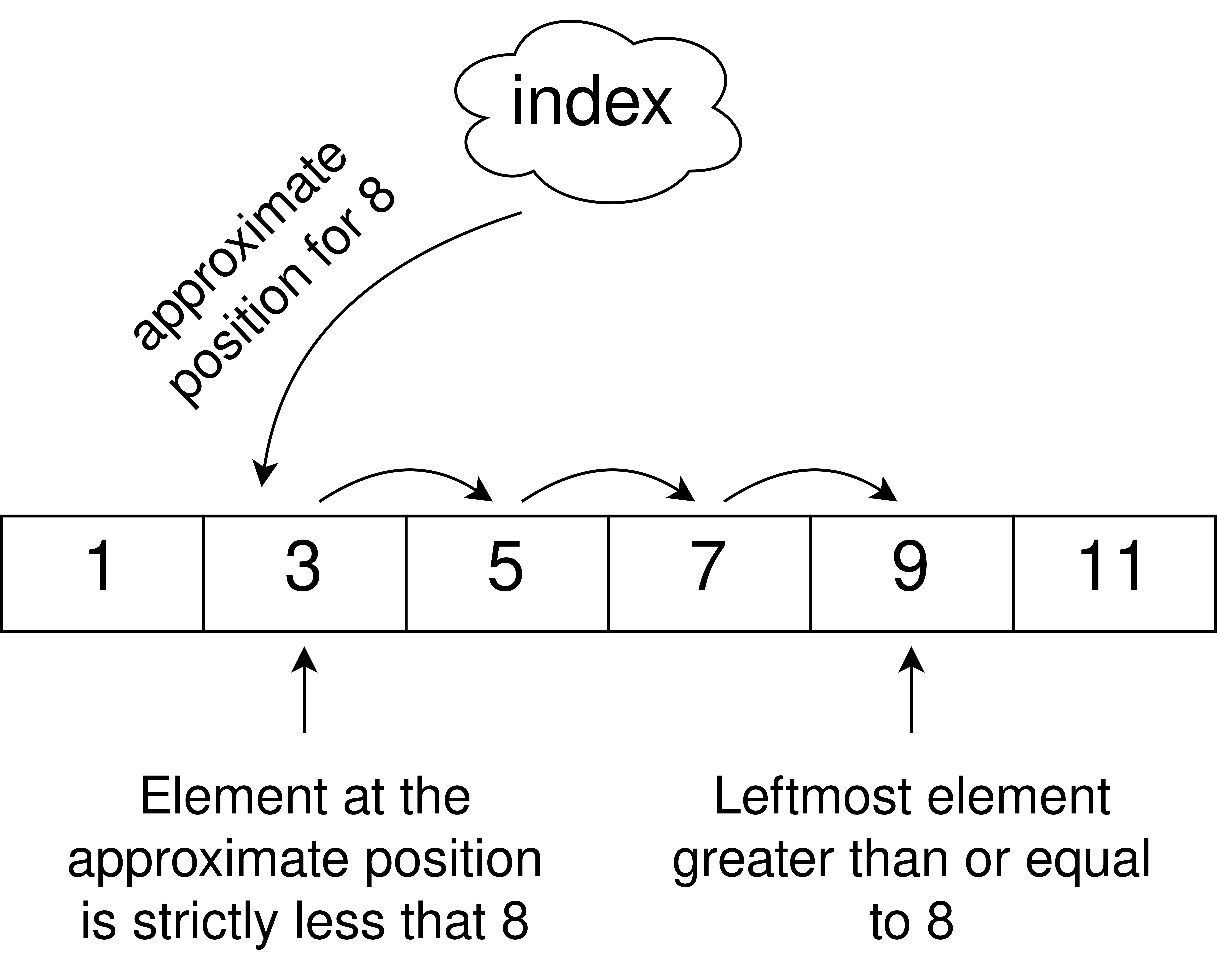}
          \label{search-right-pic}
     \end{subfigure}
    \hfill
     \begin{subfigure}[b]{0.45\linewidth}
          \centering
          \caption{Searching for the key on the left to the approximate position.}
          \includegraphics[width=\linewidth]{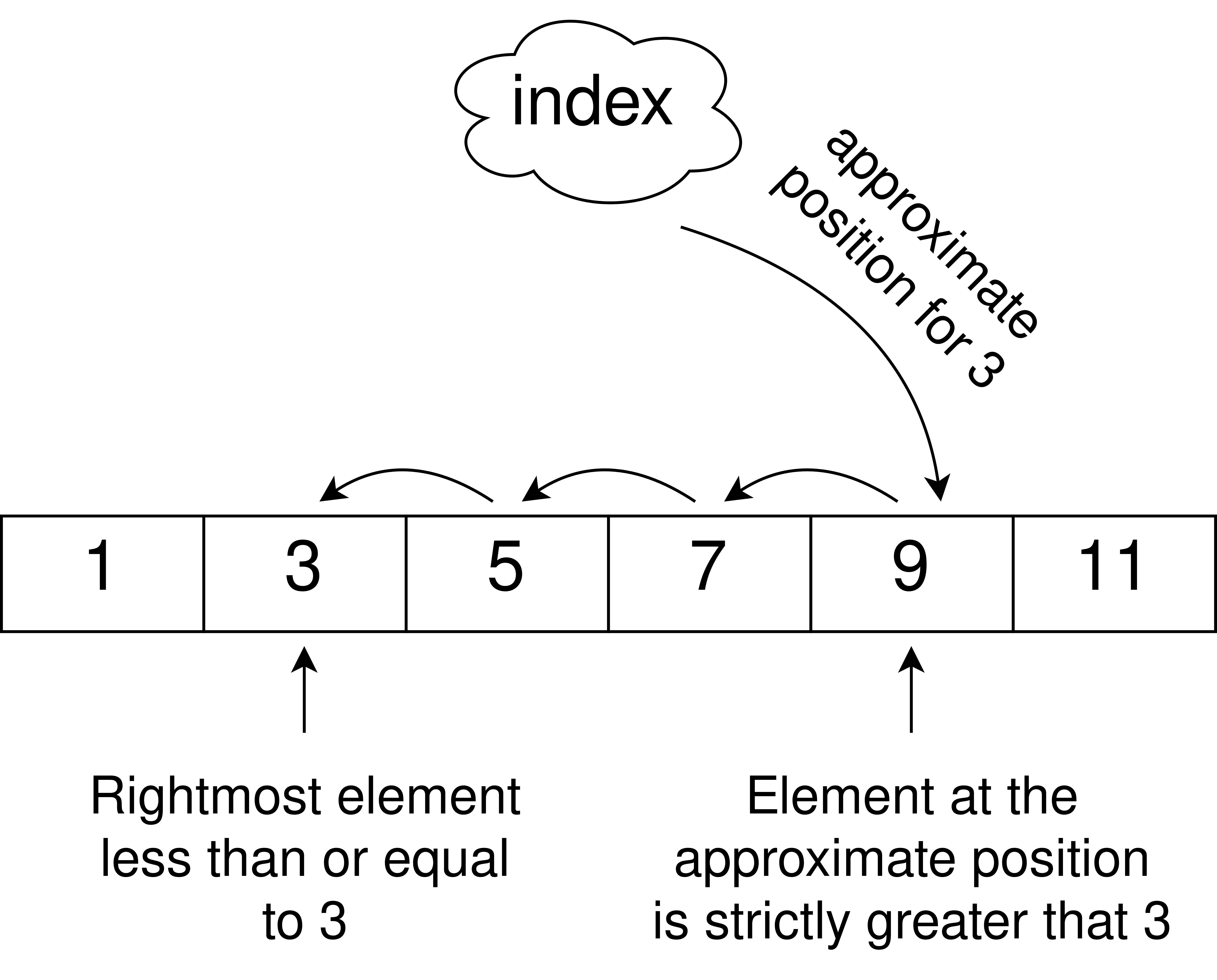}
          \label{search-left-pic}
     \end{subfigure}
    \label{search-linear-pic}
\end{figure}

Note, we can use more complex techniques instead of the linear search, e.g., exponential search~\cite{bentley1976almost}. However, they are often unnecessary, since the index usually provides an approximation good enough to finish the search only in a couple of operations. We can also use a machine learning model as an approximate index structure~\cite{kraska2018case}.

\subsection{Search in IST}
\label{search-sequential-section}

Suppose we want to find a \texttt{key} in an IST. The search algorithm is iterative: on each iteration we look for the \texttt{key} in a subtree of some node \texttt{v}. To look for the key in the whole IST we begin the algorithm with setting \texttt{v := IST.Root}, since the IST is the subtree of its root.

To find \texttt{key} in \texttt{v} subtree, we do the following (here \texttt{k} is the length of \texttt{v.Rep}):

\begin{enumerate}
    \item If \texttt{v} is empty, we conclude that \texttt{key} is not there;

    \item If \texttt{key} is found in \texttt{v.Rep} array, then we have found the key;

    \item If \texttt{key < v.Rep[0]}, the \texttt{key} can be found only in \texttt{v.C[0]} subtree (as explained in Section~\ref{IST-section}, all such keys are stored in the leftmost child subtree).
    Thus, we set \texttt{v $\leftarrow$ v.C[0]} and continue our search in the leftmost child;

    \item If \texttt{key > v.Rep[k - 1]}, the \texttt{key} can be found only in \texttt{v.C[k]} subtree (as explained in Section~\ref{IST-section}, all such keys are stored in the rightmost child subtree).
    Thus, we set \texttt{v $\leftarrow$ v.C[k]} and continue our search in the rightmost child;

    \item Otherwise, we find \texttt{j} such that \texttt{v.Rep[j - 1] < key < v.Rep[j]}. As stated in Section~\ref{IST-section} \texttt{key} can be found only in \texttt{v.C[j]}. Thus, we set \texttt{v $\leftarrow$ v.C[j]} and continue our search in the \texttt{j}-th child.
\end{enumerate}

We can implement the aforementioned algorithm the following way (Listing~\ref{contains-sequential-listing}):

\begin{lstlisting}[caption={An algorithm to search a key in an IST},escapeinside={(*}{*)},captionpos=t,numbers=none,label={contains-sequential-listing}]
fun IST.Contains(key):
    v := IST.Root
    while true:
        if v = nil: return false
        k := (*$\vert$*) v.Rep (*$\vert$*)
        if key < v.Rep[0]: v (*$\leftarrow$*) v.C[0]
        elif key > v.Rep[k - 1]: v (*$\leftarrow$*) v.C[k]
        else:
            j := interpolation_search(v.Rep, key)
            if key = v.Rep[j]: return true
            v (*$\leftarrow$*) v.C[j]
\end{lstlisting}

\subsection{Executing update operations and maintaining balance}
\label{insert-balancing-sequenctial-section}

The algorithm for inserting a key into IST is very similar to the search algorithm above. To execute \texttt{IST.Insert(key)} we do the following (Fig.~\ref{ist-single-insert-pic}):

\begin{enumerate}
    \item Initialize \texttt{v := IST.Root}: we insert the desired key in the subtree of \texttt{v}, which is initially the whole IST;
    
    \item For the current node \texttt{v}, if \texttt{key} appears in \texttt{v.Rep} array, we finish the operation~--- the key already exists.
        
    \item If \texttt{v} is a leaf and \texttt{key} does not exists in \texttt{v.Rep}, insert \texttt{key} into \texttt{v.Rep} keeping it sorted (we can find the insertion position for the new key as described in Section~\ref{interpolation-search-section});
        
    \item If \texttt{v} is an inner node and \texttt{key} does not exists in \texttt{v.Rep}, determine in which child the insertion should continue (we do it the same way as in Section~\ref{search-sequential-section}), set \texttt{v $\leftarrow$ v.C[next\_child\_idx]} and go to step~(2).
\end{enumerate}

\begin{figure}[H]
  \centering
  \caption{Insert \texttt{15}: proceed from the root to the second child since $2 < 15 < 50$
  and then to the first child since $15 < 20$} 
  \label{ist-single-insert-pic}
  \includegraphics[width=\linewidth]{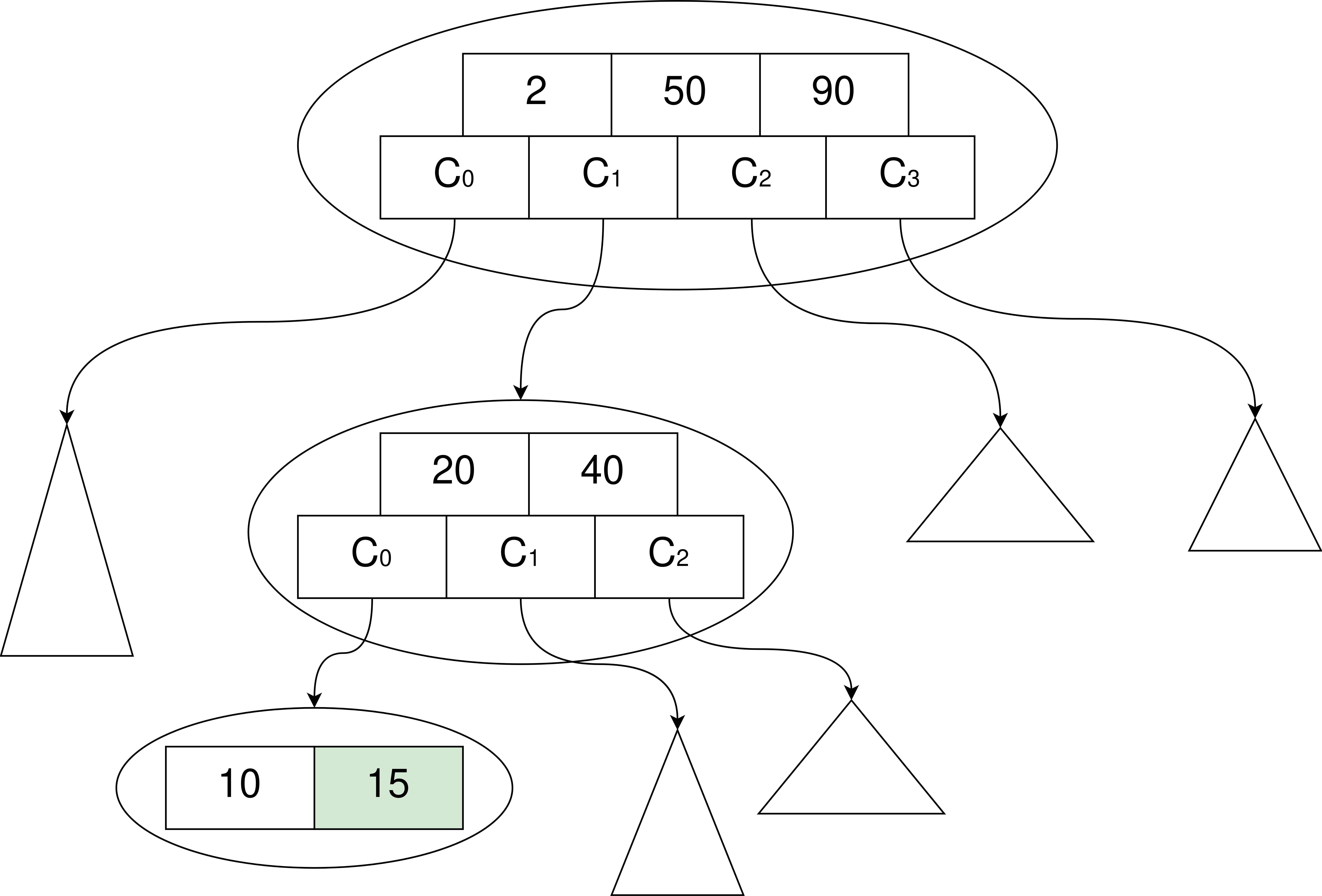}
\end{figure}

To remove a key from IST we introduce \texttt{Exists} array in each node that shows whether the corresponding key in \texttt{Rep} is in the set or not. Thus, we just need to mark a key as removed without physically deleting it. We have to take into account such marked keys during contains and insert operations. The removal algorithm is discussed in more detail in~\cite{mehlhorn1993dynamic,brown2020non,prokopec2020analysis}.


The problem with these update algorithms is that all the new keys may be inserted to a single leaf, making the IST unbalanced (Fig.~\ref{ist-unbalanced-pic}).  In order to keep the execution time low, we should keep the tree balanced. 

\begin{figure}[H]
  \centering
  \caption{Example of an unbalanced IST where the left leaf stores a single key while the right leaf stores many keys.} 
  \label{ist-unbalanced-pic}
  \includegraphics[width=\linewidth]{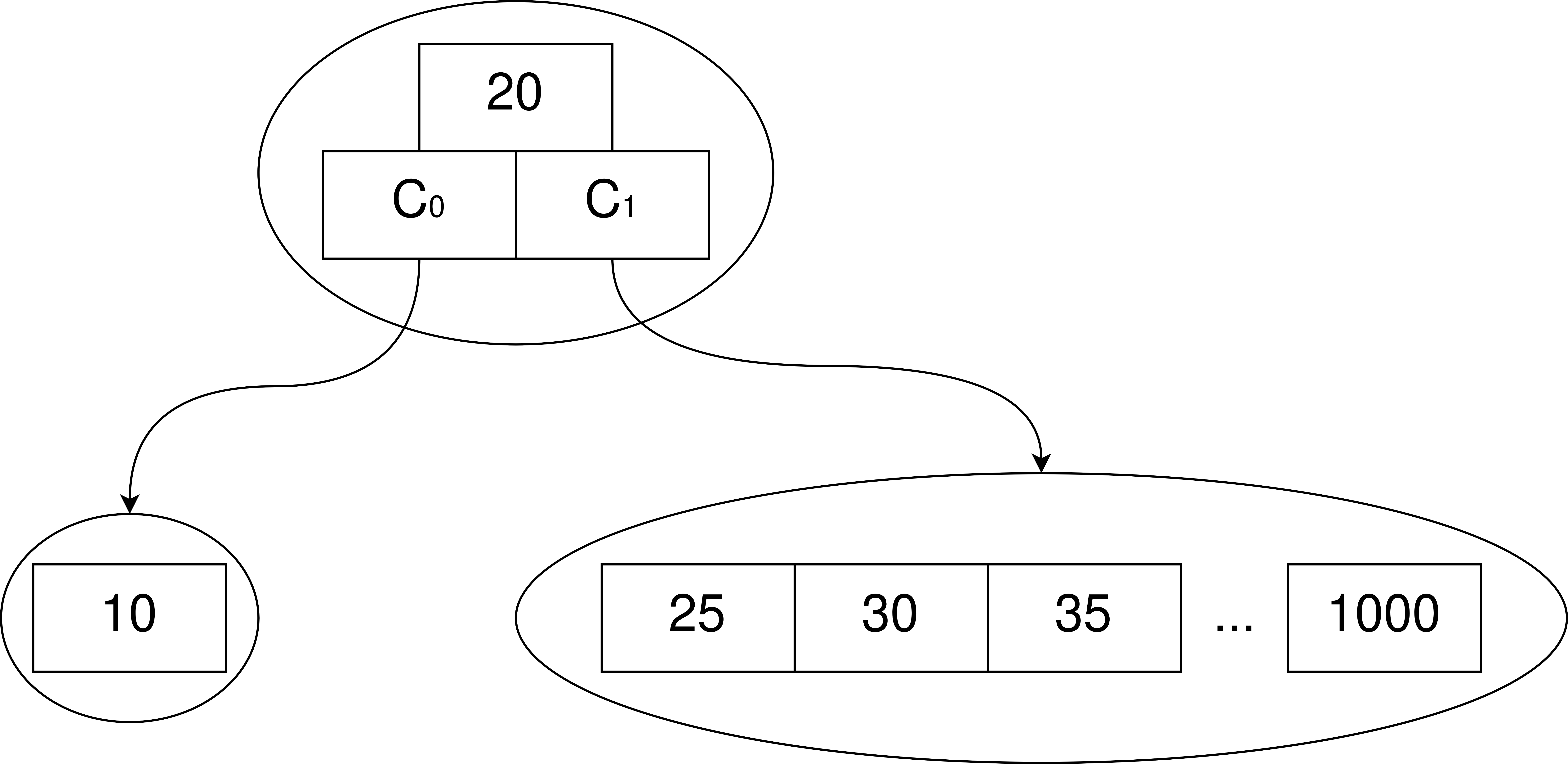}
\end{figure}

\begin{definition}
\label{ideal-definition}
\normalfont Suppose \texttt{H} is some small integer constant, e.g., \texttt{10}. An IST \texttt{T}, storing keys $x_0 < x_1 < \ldots < x_{n-1}$, is said to be \emph{ideally balanced} if either : 

\begin{itemize}
    \item \texttt{T} is a leaf IST and \texttt{n $\leq$ H};
    \item \texttt{T} is a non-leaf IST and (here \texttt{k} is the size of the of the \texttt{Rep} array):
    \begin{itemize}
        \item \texttt{n > H};
        \item $k \in \Theta(\sqrt{n})$;
        \item elements in \texttt{Rep} are equally spaced: \texttt{Rep[0]} must equal to $x_{\lfloor\frac{n}{k}\rfloor}$, \texttt{Rep[1]} must equal to $x_{2 \cdot \lfloor\frac{n}{k}\rfloor}$, \texttt{Rep[2]} must equal to $x_{3 \cdot \lfloor\frac{n}{k}\rfloor}$, and so on: \texttt{Rep[i]} must equal to $x_{(i + 1) \cdot \lfloor\frac{n}{k}\rfloor}$;
        \item all child ISTs $\{C_i\}_{i=0}^k$ are ideally balanced;
    \end{itemize}
\end{itemize}
\end{definition}

Consider an ideally balanced IST storing $n$ keys (Fig.~\ref{IST-ideal-pic}). As stated in~\cite{mehlhorn1993dynamic}:
\begin{itemize}
    \item The root of IST contains $\Theta(\sqrt{n}) = \Theta(n^{\frac{1}{2}})$ keys in its \texttt{Rep} array;
    \item Any node on the second level has \texttt{Rep} array of size $\Theta(\sqrt{n^\frac{1}{2}}) = \Theta(n^{\frac{1}{4}})$;
    \item Any node on the third level of has \texttt{Rep} array of size $\Theta(\sqrt{n^\frac{1}{4}}) = \Theta(n^{\frac{1}{8}})$ keys in its \texttt{Rep} array;
\end{itemize}

Generally, any node on the $t$-th level has \texttt{Rep} array of size $\Theta(n^{\frac{1}{2^t}})$.
Thus, an ideally balanced IST with $n$ keys has a height of $O(\log \log n)$.

\begin{figure}[H]
  \centering
  \caption{Height of an ideal IST} 
  \label{IST-ideal-pic}
  \includegraphics[width=\linewidth]{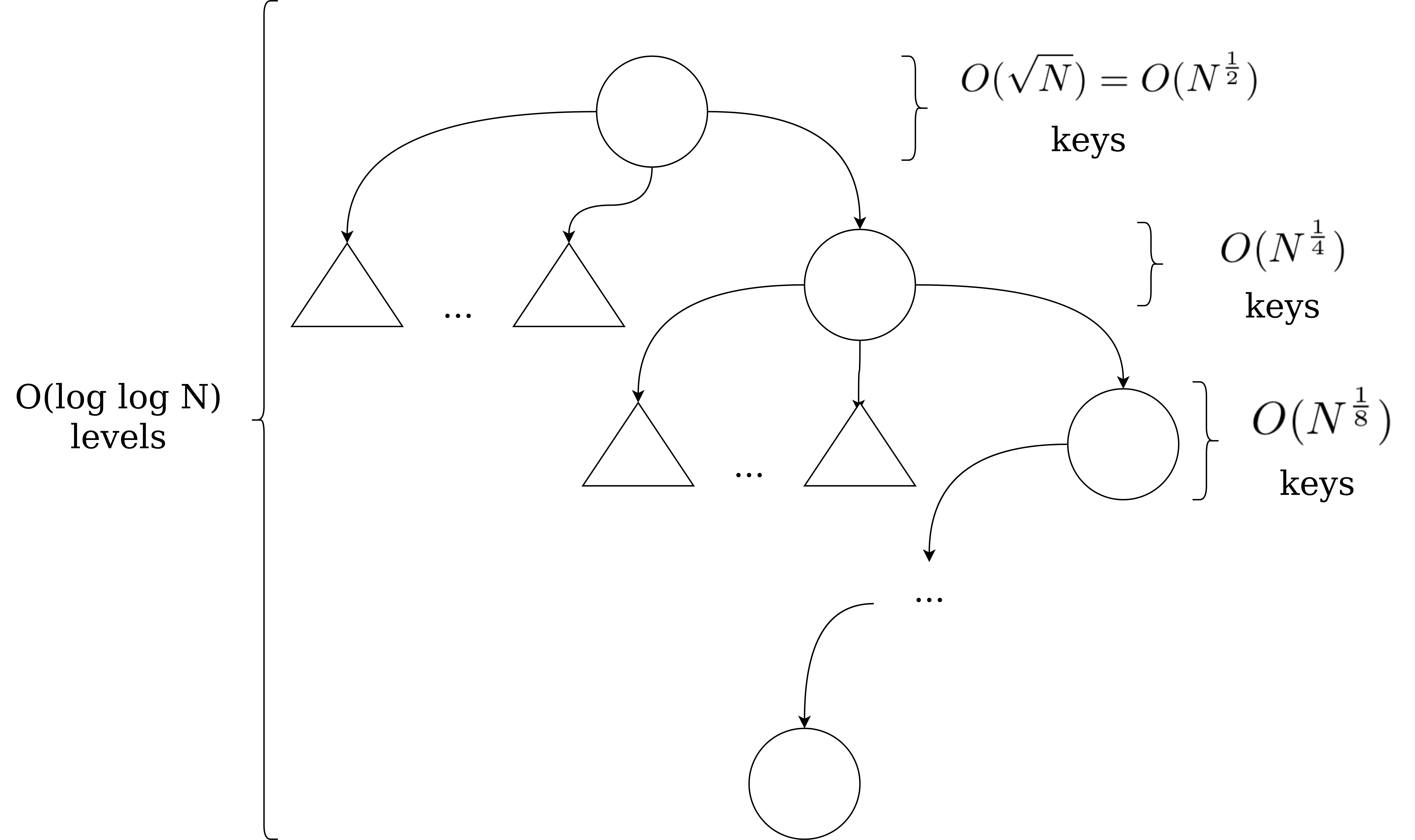}
\end{figure}

In order to keep IST balanced we maintain the number of modifications (both insertions and removals) applied to each subtree \texttt{T}. When the number of modifications applied to some subtree \texttt{T} exceeds the initial size of \texttt{T} multiplied by some constant \texttt{C}, we rebuild \texttt{T} from scratch making it ideally balanced. This rebuilding approach has a proper amortized bounds and is adopted from papers about IST~\cite{mehlhorn1993dynamic,brown2020non,prokopec2020analysis}.

\subsection{Time and space complexity}
\label{time-space-section}



Mehlhorn and Tsakalidis~\cite{mehlhorn1993dynamic} give a definition of a \emph{smooth probability distribution}. For example, uniform distribution $U(a; b)$ is smooth for every $-\infty < a < b < +\infty$.

\begin{definition}
\normalfont An insertion is \emph{$\mu$-random} if the key to be inserted is drawn from probability density $\mu$.

\normalfont A removal is \emph{random} if every key present in the set is equally likely to be removed. 
\end{definition}

Mehlhorn and Tsakalidis~\cite{mehlhorn1993dynamic} state that if probability distribution $\mu$ is smooth:

\begin{itemize}
    \item IST with $n$ keys takes $O(n)$ space;
    \item The expected amortized insertion and removal cost is $O(\log \log n)$;
    \item The amortized insertion and removal cost is $O(\log n)$;
    \item The expected search time on IST, generated by $\mu$-random insertions and random removals, is $O(\log \log n)$;
    \item the worst-case search time is $O(\log^2 n)$;
\end{itemize}

These estimations are based on the fact, that the expected cost of the interpolation search on a sorted array, generated by $\mu$-random insertions and random removals given that $\mu$ is smooth, is $O(1)$, and the ideally balanced IST storing $n$ keys has depth $O(\log \log n)$. 

Therefore, IST can execute operations in $o(\log n)$ time under reasonable assumptions. As our goal, we want to design a parallel-batched version of the IST that processes operations asymptotically faster than known sorted set implementations (e.g., red-black trees).

\section{Parallel-batched Contains}
\label{parallel-contains-section}

In this section, we describe the implementation of \texttt{ContainsBatched(keys[])} operation, introduced in Section~\ref{parallel-batched-section}.
We suppose that \texttt{keys} array is sorted.
For simplicity, we assume that IST does not support removals. In Section~\ref{parallel-remove}, we elaborate on how to support removal operations and what changes should we do in the \texttt{ContainsBatched} implementation in order to support removals.

We implement \texttt{ContainsBatched} operation the following way. At first, we introduce a function \texttt{BatchedTraverse(node, keys[], left, right, result[])}. The purpose of this function is to determine for each index $\texttt{left} \leq \texttt{i} < \texttt{right}$, whether \texttt{keys[i]} is stored in the \texttt{node} subtree. If so, set \texttt{result[i] = true}, otherwise, \texttt{result[i] = false}. 
Given the operation \texttt{BatchedTraverse}, we can implement \texttt{ContainsBatched} with almost zero effort (Listing~\ref{contains-batched-listing}):

\begin{lstlisting}[caption={Implementation of \texttt{ContainsBatched} on top of \texttt{BatchedTraverse routine}},escapeinside={(*}{*)},captionpos=t,numbers=none,label={contains-batched-listing}]
fun ContainsBatched(keys[]):
    result[] := [array of size |keys|]
    // search for all keys in the root subtree (i.e., in the whole IST)
    BatchedTraverse(IST.Root, keys, 0, |keys|, result)
    return result
\end{lstlisting}

Now, we describe \texttt{BatchedTraverse(node, keys[], left, right, result[])} implementation in leaf IST nodes in Section~\ref{leaf-traverse-section} and in non-leaf IST nodes in Section~\ref{traverse-inner-section}.

\subsection{\texttt{BatchedTraverse} in a leaf node}
\label{leaf-traverse-section}

If \texttt{node} is a leaf node, we determine for each key in \texttt{keys[left..right)} whether it exists in \texttt{node.Rep}.
Since \texttt{node} is a leaf, keys cannot be found anywhere else in \texttt{node} subtree. 




We may use \texttt{Rank} function to find the \emph{rank} of each element of \texttt{keys[left..right)} in \texttt{node.Rep} and, thus, determine for each key whether it exists in \texttt{node.Rep} (Fig.~\ref{leaf-contains-pic}). As presented in Section~\ref{primitives-section}, ranks of all keys from subarray \texttt{keys[left..right)} in \texttt{node.Rep} may be computed in parallel in $O(right - left + \vert node.Rep \vert)$ work and $O(\log^2 \left(right - left + \vert node.Rep \vert) \right)$ span.

Let us denote the rank of the key being searched \texttt{r := RankElem(node.Rep, key)}. Note that \texttt{r} equals to \texttt{{$\vert \{ \texttt{x} \in \texttt{node.Rep} : \texttt{x} \leq \texttt{key} \}\vert$}}~--- the number of elements of \texttt{node.Rep}, that are less than or equal to the \texttt{key}. 

\begin{itemize}
   \item If \texttt{r} equals zero, it means that all the elements of \texttt{node.Rep} are strictly greater than \texttt{key}. Thus, \texttt{key} is not in \texttt{node.Rep} and the result is \texttt{false};

   \item Otherwise, \texttt{r} elements of \texttt{node.Rep} are less than or equal to \texttt{key}. In that case, we check \texttt{node.Rep[r - 1]}: if it equals \texttt{key}, the key is found in \texttt{node.Rep}, otherwise, the key is not found (Fig.~\ref{leaf-contains-pic}).
\end{itemize}

\begin{figure}[H]
  \centering
  \caption{Execution of \texttt{BatchedTraverse} in an IST leaf. Here \texttt{Rank(node.Rep, keys[left..right)) = [0, 1, 2, 2, 3, 4]}.} 
  \label{leaf-contains-pic}
  \includegraphics[width=\linewidth]{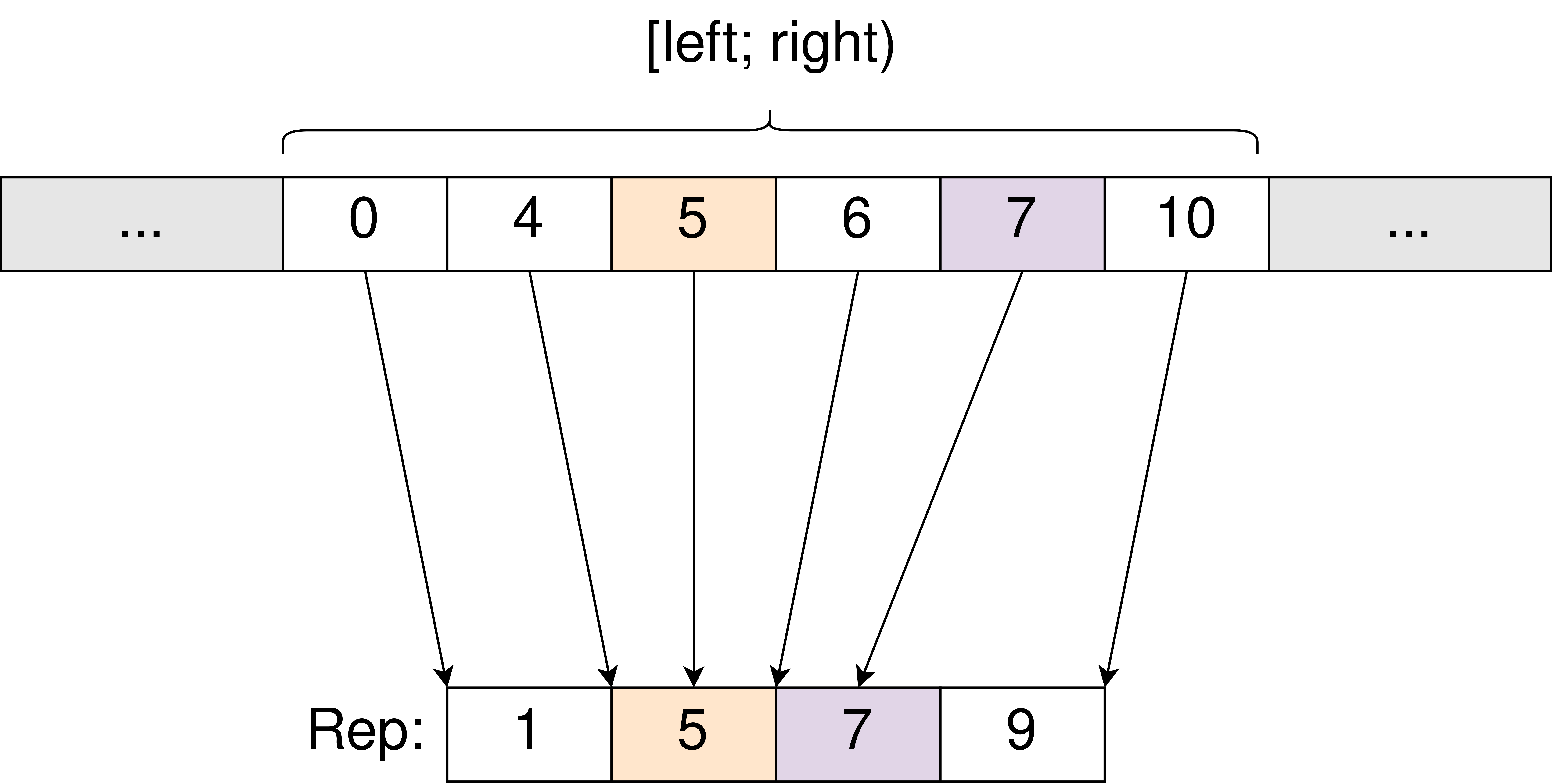}
\end{figure}

We can check the existence of all the keys and set all the results in parallel using parallel loop the following way (Listing~\ref{rank-in-leaf-listing}):
\begin{lstlisting}[caption={Using \texttt{Rank} to find keys in a leaf node in parallel},escapeinside={(*}{*)},captionpos=t,numbers=none,label={rank-in-leaf-listing}]
rank := Rank(node.Rep, keys[left..right))
pfor i in left..right:
    r := rank[i - left]
    if r = 0 or node.Rep[r - 1] (*$\neq$*) keys[i]:
        result[i] (*$\leftarrow$*) false
    else:
        result[i] (*$\leftarrow$*) true
\end{lstlisting}

\subsection{\texttt{BatchedTraverse} in an inner node}
\label{traverse-inner-section}

Consider now the \texttt{BatchedTraverse} procedure on an inner node (Fig.~\ref{inner-contains-pic}).

\begin{figure}[H]
  \centering
  \caption{Execution of \texttt{BatchedTraverse} in an inner node of an IST.} 
  \label{inner-contains-pic}
  \includegraphics[width=\linewidth]{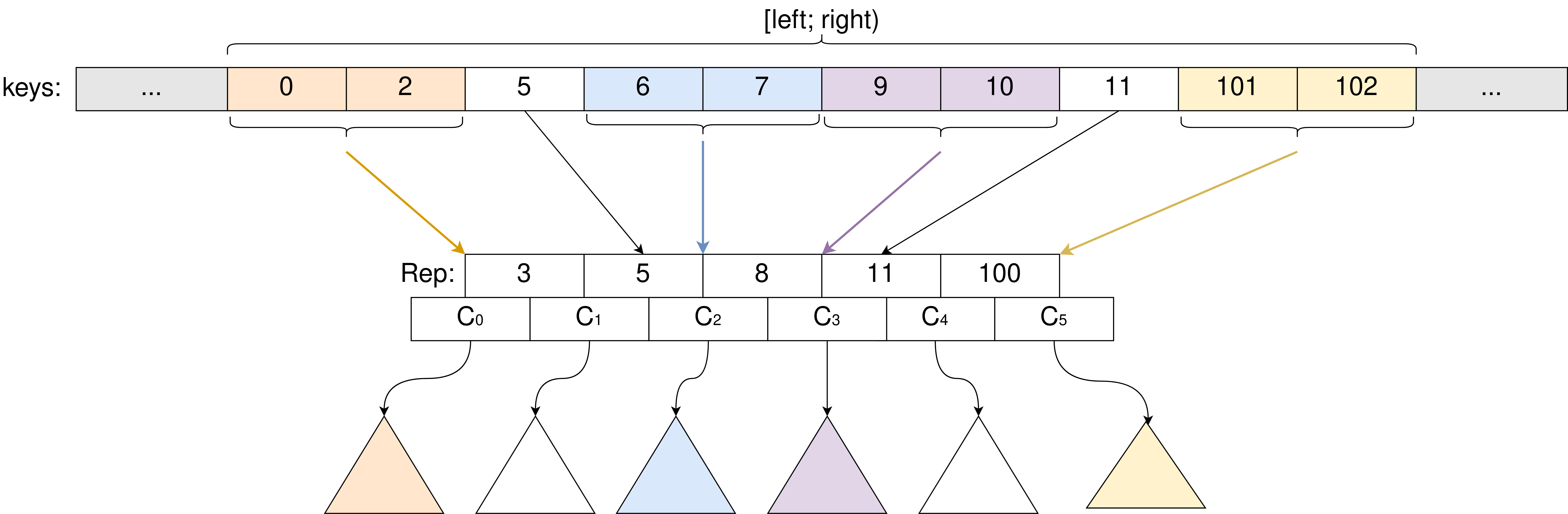}
\end{figure}

We begin its execution with finding the position for each key from \texttt{keys[left..right)} in \texttt{node.Rep} (i.e., the number of \texttt{node.Rep} elements less than each given key). We may do it using \texttt{Rank} function as in Section~\ref{leaf-traverse-section}~--- it will take $O(right - left + \vert node.Rep \vert)$ work and $O(\log^2 \left(right - left + \vert node.Rep \vert) \right)$ span. 

However, we can also use the interpolation search (described in Section~\ref{interpolation-search-section}) and parallel loop (described in Section~\ref{primitives-section}) for that purpose: see Listing~\ref{search-in-leaf-listing}.

\begin{lstlisting}[caption={Using interpolation search to find keys in IST node in parallel},escapeinside={(*}{*)},captionpos=t,numbers=none,label={search-in-leaf-listing}]
key_indexes := [array of size (right - left)]
pfor i in left..right:
    key_indexes[i - left] (*$\leftarrow$*) interpolation_search(node.Rep, keys[i])
\end{lstlisting}

Denoting $S$ as an interpolation search time in \texttt{node.Rep}, this algorithm takes $O((right - left) \cdot S)$ work and $O(\log \left(right - left \right) \cdot T)$ span. As stated in Section~\ref{time-space-section}, $S$ is expected to be $O(1)$ if the IST is obtained using $\mu$-random insertions and random removals given that $\mu$ is \emph{smooth}. Thus, under this assumptions interpolation-based traverse is expected to be faster than \texttt{Rank}-based one (from Section~\ref{leaf-traverse-section}) when \texttt{node.Rep} array is big.

Some keys of the input array \texttt{keys[left..right)} (e.g., \texttt{5} and \texttt{11} in Fig.~\ref{inner-contains-pic}) are found in the \texttt{Rep} array. For such keys, we set \texttt{Result[i] $\leftarrow$ true} (indeed, they exist in \texttt{Rep} array, therefore they exist in \texttt{node} subtree). All other keys can be divided into three categories:

\begin{itemize}
    \item Keys that are strictly less than \texttt{Rep[0]} (e.g., \texttt{0} and \texttt{2} in Fig.~\ref{inner-contains-pic}) can only be found in \texttt{C[0]} subtree (as explained in Section~\ref{IST-structure-section}). Therefore, for such keys we should continue the traversal in \texttt{C[0]}.

    \item Keys that are strictly greater than \texttt{Rep[k - 1]} (e.g., \texttt{100} and \texttt{101} in Fig.~\ref{inner-contains-pic}) can only be found in \texttt{C[k]} (as explained in Section~\ref{IST-structure-section}). Therefore, for such keys we continue the traversal in \texttt{C[k]}.

    \item Keys that lie between \texttt{Rep[i]} and \texttt{Rep[i + 1]} for some $i \in [0; k - 2]$ (e.g., \texttt{6} and \texttt{7} for \texttt{i = 1} or \texttt{9} and \texttt{10} for \texttt{i = 2} in Fig.~\ref{inner-contains-pic}) can only be found in \texttt{C[i + 1]} (as explained in Section~\ref{IST-structure-section}). Therefore, for such keys we continue the traversal in \texttt{C[i + 1]}.
\end{itemize}

Note that some child nodes (e.g., \texttt{C[1]} and \texttt{C[4]} in Fig.~\ref{inner-contains-pic}) are guaranteed not to contain any key from \texttt{keys[left..right)} thus we do not continue the search in such nodes.

After determining in which child the search of each key from \texttt{keys[left..right)} should continue we proceed to searching for keys in children in parallel: for example in Figure~\ref{inner-contains-pic} we would continue searching:

\begin{itemize}
    \item For \texttt{keys[0..2)} in \texttt{C[0]}, thus we execute \texttt{BatchedTraverse(node.C[0], keys[], 0, 2, result[])};
    \item For \texttt{keys[3..5)} in \texttt{C[2]}, thus we execute \texttt{BatchedTraverse(node.C[2], keys[], 3, 5, result[])};
    \item For \texttt{keys[5..7)} in \texttt{C[3]}, thus we execute \texttt{BatchedTraverse(node.C[3], keys[], 5, 7, result[])};
    \item For \texttt{keys[8..10)} in \texttt{C[5]}, thus we execute \texttt{BatchedTraverse(node.C[5], keys[], 8, 10, result[])};
\end{itemize}

All these \texttt{BatchedTraverse} calls can be done in parallel, since there is no dependencies between them.




\section{Parallel-batched Insert}
\label{parallel-insert}

We now consider the implementation of the operation \texttt{InsertBatched(keys[])}, introduced in Section~\ref{parallel-batched-section}.
Again, we suppose that array \texttt{keys[]} is sorted.
For simplicity, we consider \texttt{InsertBatched} implementation on an IST without removals. In Section~\ref{parallel-remove}, we explain how what changes should we do in the \texttt{InsertBatched} implementation to support remove operations.

For simplicity now we consider batch-insertion into unbalanced IST. In Section~\ref{parallel-rebuilding-section} we elaborate on how to make the IST balanced.

We begin the insertion procedure with filtering out keys already present in the IST since there is no point in inserting such keys to the set. 
We can do this using the \texttt{ContainsBatched} routine (see Section~\ref{parallel-contains-section}) together with the \texttt{Filter} primitive (see Section~\ref{primitives-section}): we filter out all the keys for which \texttt{ContainsBatched} returns \texttt{true}. For example, if the IST contains keys \texttt{1, 3, 5, 7, 9} and we are going to insert keys \texttt{[2, 4, 5, 7, 8]} in it, we filter out keys \texttt{5} and \texttt{7} (since they already exist in the IST) and insert only keys \texttt{[2, 4, 8]}.

We implement our procedure recursively in the same way as \texttt{BatchedTraverse}. 
Note, that each \texttt{key} being inserted is not present in IST, thus for each \texttt{key} our traversal finishes in some leaf (Fig.~\ref{insert-leaves-pic}) since we cannot find that \texttt{key} in some \texttt{Rep} array on the traversal path.

\begin{figure}[H]
    \centering
    \caption{Inserting a batch of keys in the IST} 
    \label{insert-leaves-pic}
    \includegraphics[width=\linewidth]{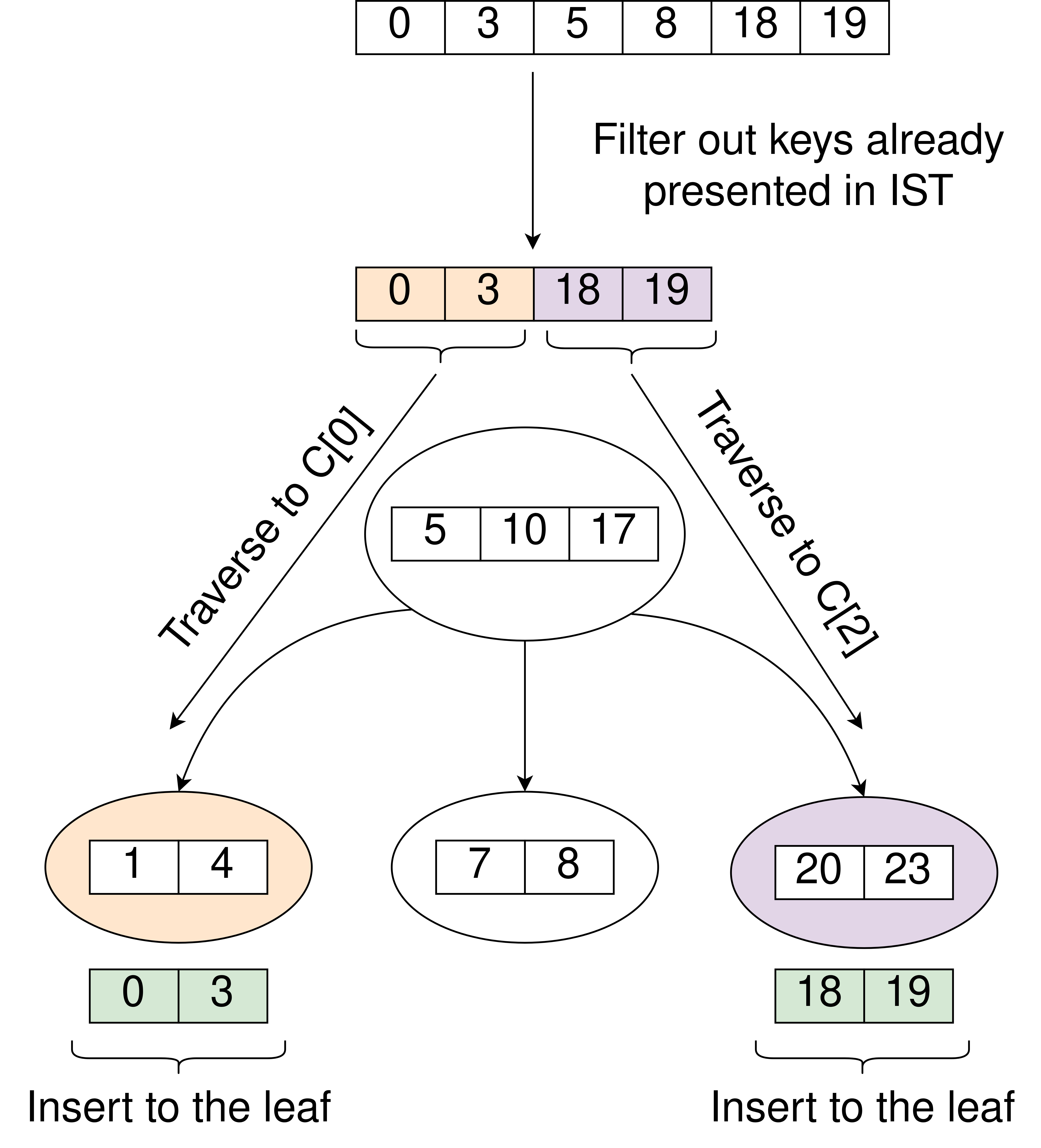}
\end{figure}


After we finish the insertion traversal, we end up with a collection of leaves $\{leaf_i\}_{i=1}^t$ and in each leaf $leaf_i$ from that collection we should insert some corresponding subarray \texttt{keys[$\texttt{left}_i \ldots \texttt{right}_i$)} of the batch being inserted. Note, that in some leaves we do not insert any keys (e.g., the leaf in the middle from Fig.~\ref{insert-leaves-pic}). For example, in Fig.~\ref{insert-leaves-pic} we insert \texttt{keys[0..2)} (i.e., \texttt{0} and \texttt{3}) to the leftmost leaf, while inserting \texttt{keys[2..4)} (i.e., \texttt{18} and \texttt{19}) to the rightmost leaf.

To finish the insertion, we just merge \texttt{keys[$\texttt{left}_i \ldots \texttt{right}_i$)} with \texttt{$\texttt{leaf}_i$.Rep} and get the new \texttt{Rep} array: \texttt{$leaf_i$.Rep $\leftarrow$ Merge(keys[$left_i \ldots right_i$), $leaf_i$.Rep)}.
Now, each target leaf $\texttt{leaf}_i$ contains all the keys that should be inserted into it.

\section{Parallel-batched Remove}
\label{parallel-remove}

We now consider the implementation of the operation \texttt{RemoveBatched(keys[])}, introduced in Section~\ref{parallel-batched-section}. Again, we suppose that array \texttt{keys} is sorted.


We begin the removal procedure with filtering out keys, not presented in the IST, since there is no point in removing such keys from the IST. We can do this using the \texttt{ContainsBatched} routine (see Section~\ref{parallel-contains-section}) together with the \texttt{Filter} primitive (see Section~\ref{primitives-section}): we filter out all the keys for which \texttt{ContainsBatched} returns \texttt{false}. For example, if the IST contains keys \texttt{1, 3, 5, 7, 9} and we are going to remove keys \texttt{[2, 3, 6, 7, 9]} from it, we filter out keys \texttt{2} and \texttt{6} (since they do not exist in the IST) and thus we remove only the keys \texttt{[3, 7, 9]}.

The remove traversal is executed similarly to the traversal from Section~\ref{parallel-contains-section} with one major difference: when we encounter \texttt{key} to be removed in \texttt{Rep} array of some node \texttt{v} we mark the key as logically removed without physically deleting it from \texttt{v.Rep} (Fig~\ref{remove-mark-pic}). 

\begin{figure}[H]
  \centering
  \caption{Removing a batch of keys from the IST} 
  \label{remove-mark-pic}
  \includegraphics[width=\linewidth]{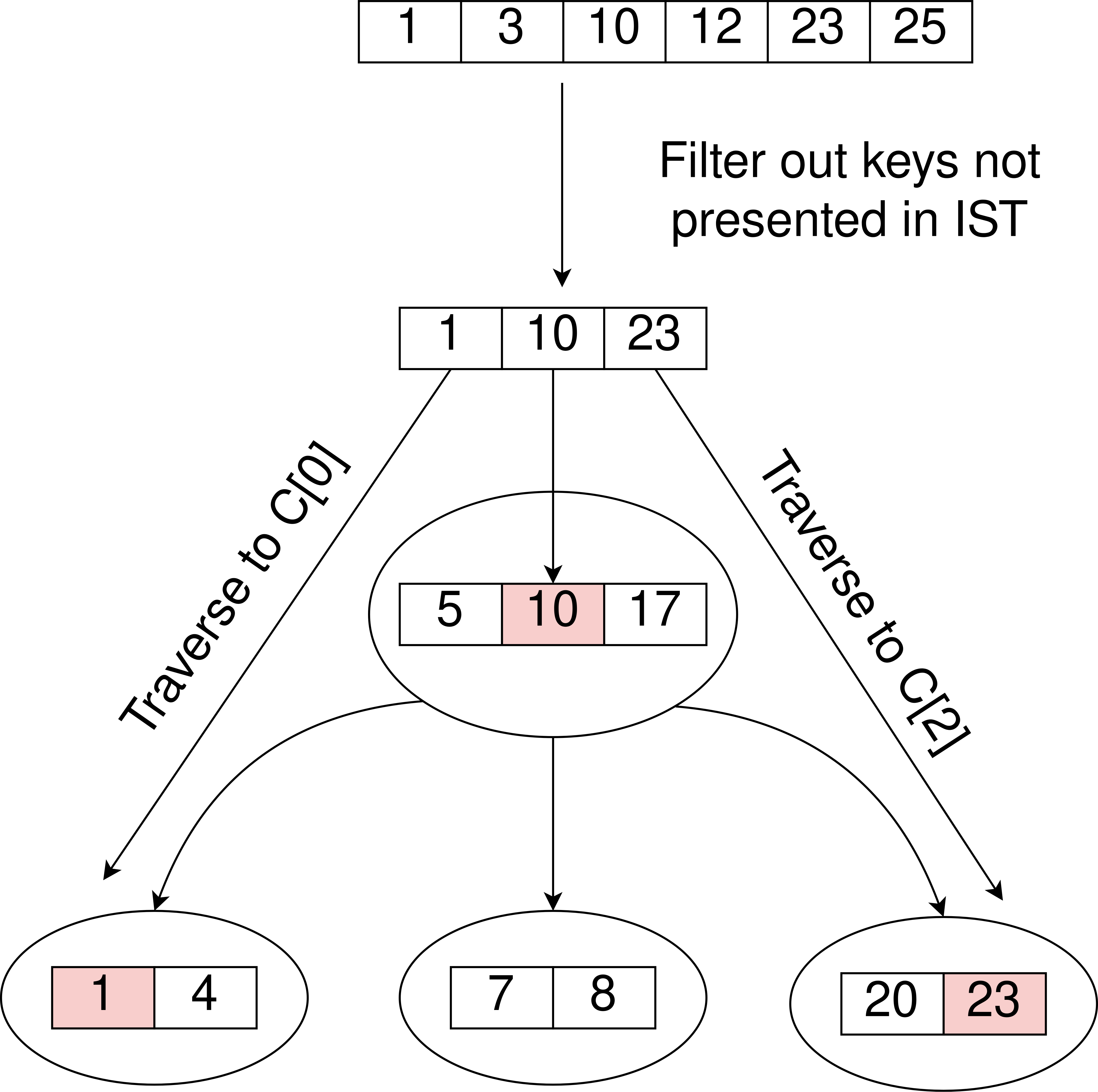}
\end{figure}

To allow marking keys as removed we augment each node \texttt{v} with an array \texttt{Exists}. \texttt{v.Exists} array has the same size as \texttt{v.Rep} (say \texttt{k}) and all its elements are initially set to \texttt{true}. When we need to remove \texttt{i}-th key from \texttt{v.Rep}, we just set \texttt{v.Exists[i] $\leftarrow$ false} (similar technique was proposed e.g., in ~\cite{mehlhorn1993dynamic,brown2020non,prokopec2020analysis}).

Note that now we may have keys, that physically exist in some \texttt{Rep} array in the tree, but such keys do not exists in the tree logically, since they have been removed by some previous \texttt{RemoveBatched} operation (see e.g., keys \texttt{1}, \texttt{10} and \texttt{23} in Fig.~\ref{remove-mark-pic}).
Since now we have a logical removal, we should modify the implementations of \texttt{ContainsBatched} (from Section~\ref{parallel-contains-section}) and \texttt{InsertBatched} (from Section~\ref{parallel-insert}) so that these operations would take into account the possible existence of such logically removed keys.

During the execution of \texttt{ContainsBatched} when we encounter the \texttt{key} being searched in the \texttt{Rep} array of some node \texttt{v} (say \texttt{v.Rep[i] = key}), we check \texttt{v.Exists[i]}: 

\begin{itemize}
    \item If \texttt{v.Exists[i] = true} then we conclude that \texttt{key} exists in the set;
    \item Otherwise, we conclude that \texttt{key} does not exist in the set, since it has been removed by previous \texttt{RemoveBatched} operation;
\end{itemize}

Now, we explain the updates to \texttt{InsertBatched} procedure. As was stated in Section~\ref{parallel-insert}, we cannot encounter any of the keys being inserted in the \texttt{Rep} array of any node of the IST, since we filter out all the keys existing in the IST before the insertion actually begins.
However, when keys can be logically removed this is not true anymore: we can encounter logically removed keys in the \texttt{Rep} array of some node \texttt{v}. Of course, such keys have the corresponding entries in \texttt{v.Exists} array set to \texttt{false}, since such keys do not logically exist in IST (Fig.~\ref{remove-insert-1-pic}). 

Suppose we are inserting \texttt{key} to the IST and we encounter it in the \texttt{Rep} array of some node \texttt{v} (i.e., \texttt{v.Rep[i] = key} for some \texttt{i}). Thus, we can just set \texttt{v.Exists[i] $\leftarrow$ true} (Fig.~\ref{remove-insert-2-pic}). This way the insert operation ``revives'' a previously removed key.

\begin{figure}
\centering
    \begin{subfigure}[b]{0.45\linewidth}
        \centering
        \caption{Keys \texttt{1}, \texttt{10} and \texttt{23} are marked as removed}
        \label{remove-insert-1-pic}
        \includegraphics[width=\linewidth]{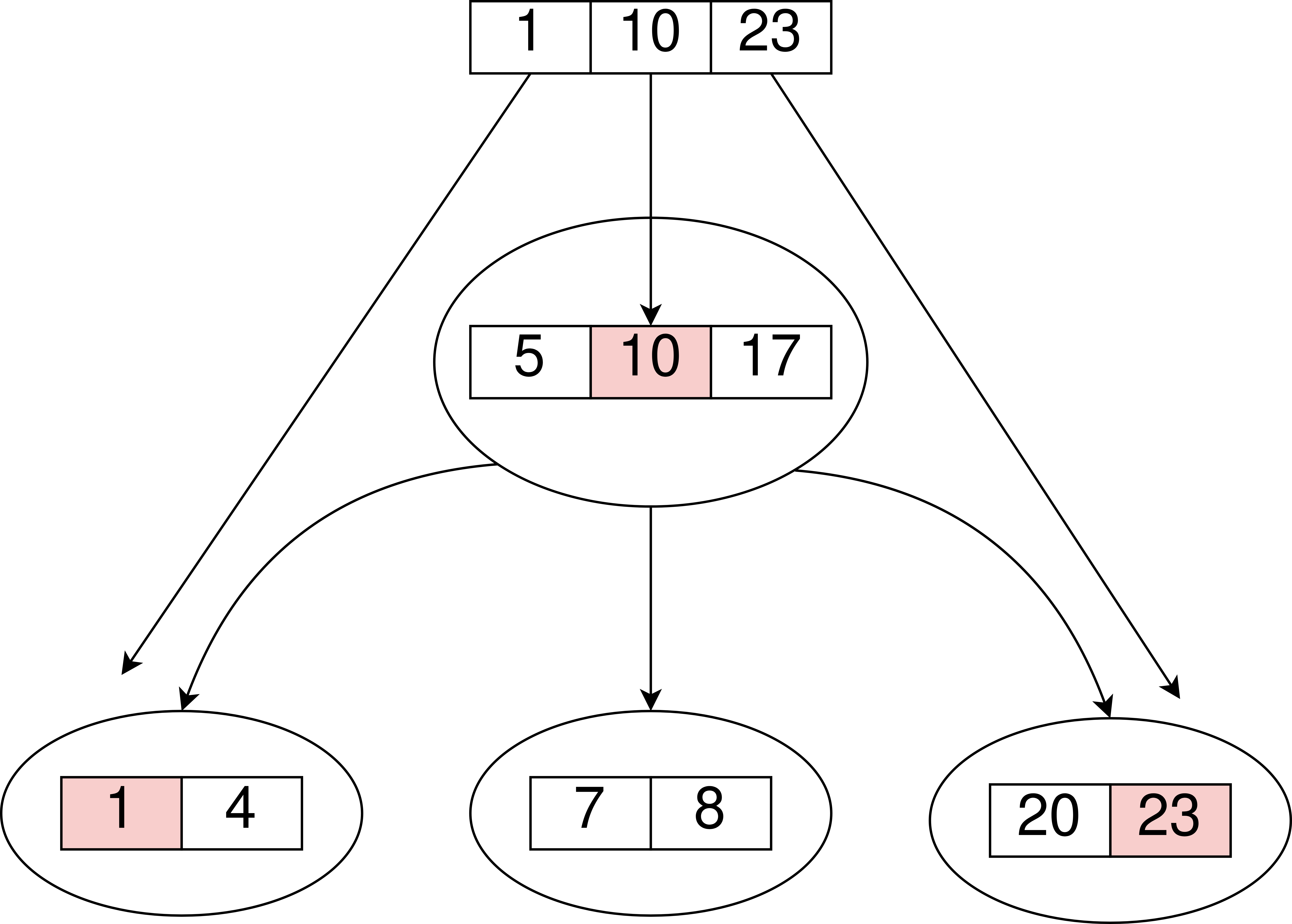}
    \end{subfigure}
    \hfill
    \begin{subfigure}[b]{0.45\linewidth}
        \centering
        \caption{Keys \texttt{10} and \texttt{23} are ``revived`` by a subsequent insert operation}
        \label{remove-insert-2-pic}
        \includegraphics[width=\linewidth]{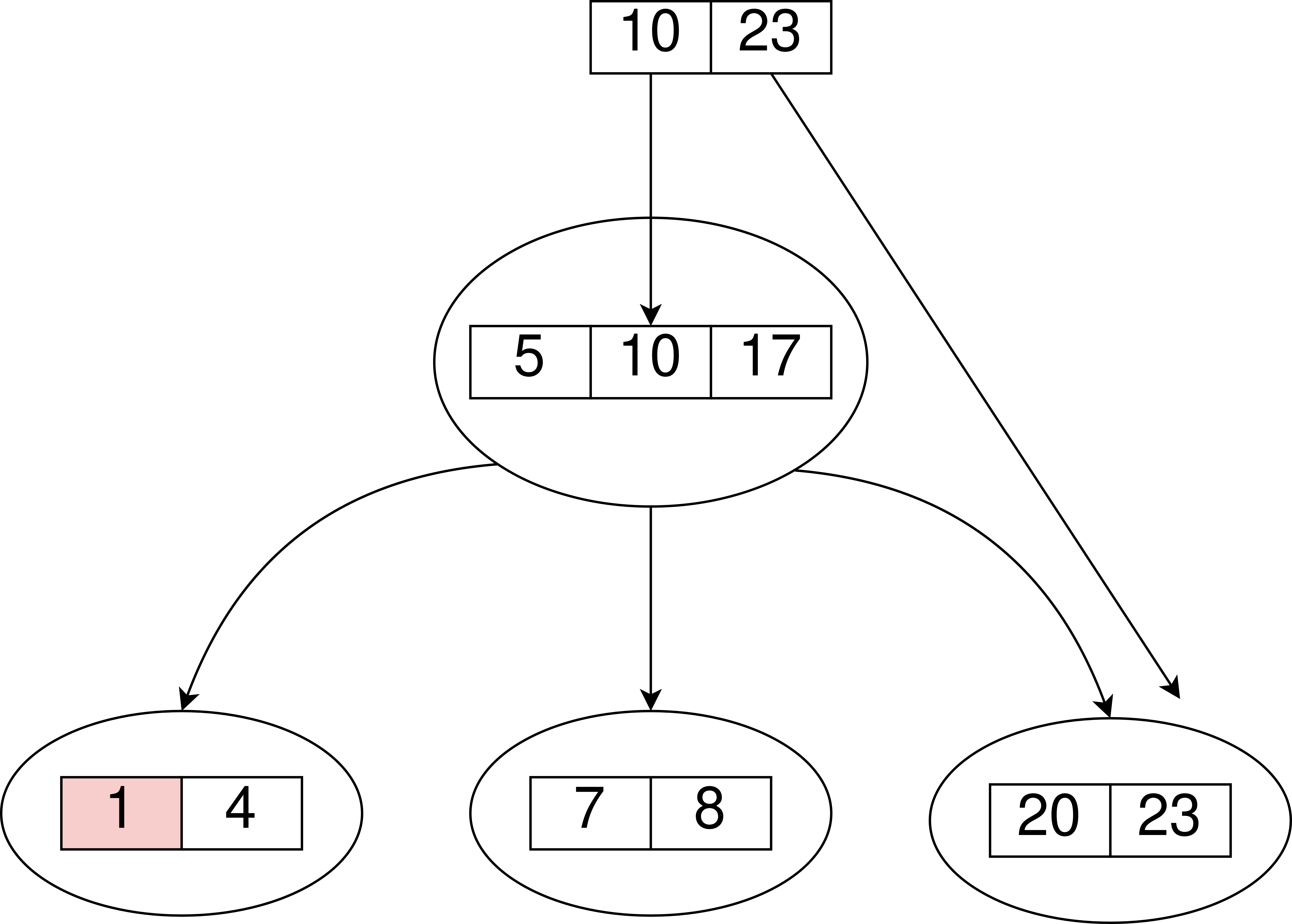}
    \end{subfigure}
\caption{Insertion of a key, that still exists in the IST physically, but is removed logically}
\label{tree-pics}
\end{figure}

\section{Parallel tree rebuilding}
\label{parallel-rebuilding-section}

\subsection{Rebuilding principle}

Following the removal algorithm from Section~\ref{parallel-remove}, we face the following problem: logically removed keys still reside in the IST wasting memory space (Fig.~\ref{empty-ist-pic}), since we do not have any mechanism for physically deleting removed keys. We can even end up in a situation when all the keys in the IST are logically removed, but the IST still occupies a great amount of memory without being able to reclaim any of the occupied space.

\begin{figure}[H]
  \centering
  \caption{IST, that is logically empty, still occupies some space for removed keys} 
  \label{empty-ist-pic}
  \includegraphics[width=\linewidth]{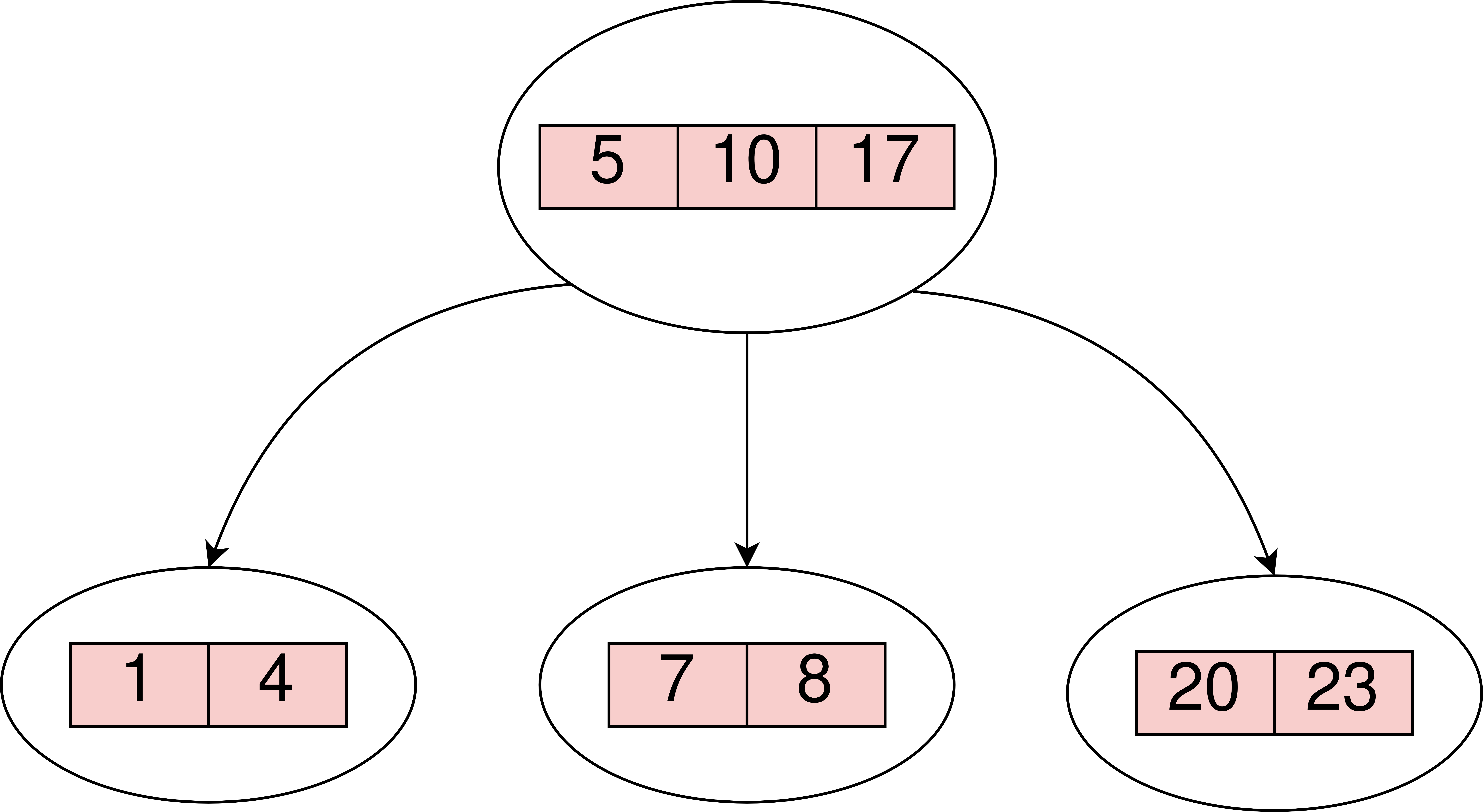}
\end{figure}

We can employ \emph{subtree rebuilding} approach to reclaim space, occupied by logically removed keys. Moreover, as stated in Section~\ref{insert-balancing-sequenctial-section}, we employ the very same approach to keep the IST balanced (and thus keep the operation execution time bounded).

Our rebuilding algorithm is conceptually similar to the algorithm proposed e.g., in~\cite{mehlhorn1993dynamic,brown2020non,prokopec2020analysis}: For each node of the IST we maintain \texttt{Mod\_Cnt}~--- the number of modifications (successful insertions and removals) applied to that node subtree. Moreover, each node stores \texttt{Init\_Subtree\_Size}~--- the number of keys in that node subtree when the node was created.

Suppose we are executing an update batch operation \texttt{Op} (i.e., either \texttt{InsertBatched} or \texttt{RemoveBatched}) in node \texttt{v} and \texttt{Op} increases \texttt{v.Mod\_Cnt} by \texttt{k} (i.e., it either inserts \texttt{k} new keys to \texttt{v} subtree or removes \texttt{k} existing keys from \texttt{v} subtree).

If \texttt{v.Mod\_Cnt + k $\leq$ C $\cdot$ v.Init\_Subtree\_Size} (where \texttt{C} is a predefined integer constant, e.g., \texttt{2}) we increment \texttt{v.Mod\_Cnt} by \texttt{k} and continue the execution of \texttt{Op} in an ordinary way (as described in Sections~\ref{parallel-insert} and ~\ref{parallel-remove}); Otherwise, we rebuild the whole subtree of \texttt{v}. The subtree rebuilding works in the following way:

\begin{enumerate}
    \item We \emph{flatten} the subtree into an array: we collect all non-removed keys from the subtree to array \texttt{subtree\_keys} in ascending order. This operation is described in more detail in Section~\ref{flatten-section};

    \item We apply the operation, that triggered the rebuilding, to the \texttt{subtree\_keys} array:

    \begin{enumerate}
        \item If \texttt{InsertBatched(v, keys[], left, right)} operation triggered the rebuilding, we merge (see Section~\ref{primitives-section}) \texttt{keys[left..right)} into \texttt{subtree\_keys}: \texttt{subtree\_keys $\leftarrow$ Merge(subtree\_keys, keys[left..right))};

        \item If \texttt{RemoveBatched(v, keys[], left, right)} operation triggered the rebuilding, we filter out \texttt{keys[left..right)} from the \texttt{subtree\_keys} via the \texttt{Difference} operation (see Section~\ref{primitives-section}): \texttt{subtree\_keys $\leftarrow$ Difference( subtree\_keys, keys[left..right))};
    \end{enumerate}

    \item Finally, we build an ideal IST \texttt{v'}, containing all entries from \texttt{subtree\_keys}. This operation is described in more detail in Section~\ref{build-section}. Note, that \texttt{v'} will contain all the keys we were inserting into \texttt{v} (if we are inserting new keys) or will not contain the keys were removing from \texttt{v} (if we are removing the existing keys).
\end{enumerate}

\subsection{Flattening an IST into an array in parallel}
\label{flatten-section}

First of all, we need to know how many keys are located in each node subtree. We store this number in a \texttt{Size} variable in each node and maintain it the following way: 

\begin{itemize}
    \item When creating new node \texttt{v}, initialize \texttt{v.Size} with the initial number of keys in \texttt{v} subtree;
    \item When inserting \texttt{m} new keys to \texttt{v}'s subtree, increment \texttt{v.Size} by \texttt{m};
    \item When removing \texttt{m} existing keys from \texttt{v}'s subtree, decrement \texttt{v.Size} by \texttt{m}.
\end{itemize}

To flatten the whole subtree of \texttt{node} we allocate an array \texttt{subtree\_keys} of size \texttt{node.Size} where we shall store all the keys from \texttt{node} subtree. We implement the flattening recursively, via the \texttt{Flatten(v, subtree\_keys[], left, right)} procedure, which fills \texttt{subtree\_keys[left..right)} subarray with all the keys from \texttt{v}'s subtree. To flatten the whole subtree of \texttt{node} into newly-allocated array \texttt{subtree\_keys} of size \texttt{node.Size} we use \texttt{Flatten(node, subtree\_keys, 0, node.Size)}.

Note that non-leaf node \texttt{v} has \texttt{2k + 1} sources of keys:

\begin{itemize}
    \item \texttt{v.C[0]} containing \texttt{v.C[0].Size} keys;

    \item \texttt{v.Rep[0]} containing \texttt{1} key if \texttt{v.Exists[0] = true} and \texttt{0} keys otherwise;

    \item \texttt{v.C[1]} containing \texttt{v.C[1].Size} keys;

    \item \texttt{v.Rep[1]} containing \texttt{1} key if \texttt{v.Exists[1] = true} and \texttt{0} keys otherwise;

    \item ...

    \item \texttt{v.Rep[k - 1]} containing \texttt{1} key if \texttt{v.Exists[k - 1] = true} and \texttt{0} keys otherwise;

    \item \texttt{v.C[k]} containing \texttt{v.C[k].Size} keys.
\end{itemize}

Here \texttt{C[i]} is \texttt{(2 $\cdot$ i)}-th key source and \texttt{Rep[i]} is \texttt{(2 $\cdot$ i + 1)}-th key source. Note that for a leaf node all children just contain zero keys. 

Now for each key source we must find the position in the \texttt{subtree\_keys} array where that source will place its keys. 
To do this we calculate array \texttt{sizes} of size \texttt{2k + 1}: \texttt{i}-th source of keys stores its keys count in \texttt{sizes[i]}. After that we calculate \texttt{positions := Scan(sizes)} to find the prefix sums of sizes array. After that 

\begin{equation*}
            \texttt{positions}[i] = 
             \begin{cases}
               0 & i = 0\\
               \sum\limits_{j = 0}^{i - 1} \texttt{sizes}[j] & i > 0
             \end{cases}
            \end{equation*}

Consider now \texttt{i}-th key source. All prior key sources will contain \texttt{positions[i]} keys in total, thus \texttt{i}-th key source should place its keys into \texttt{subtree\_keys} array starting from \texttt{left + positions[i]} position (Fig.~\ref{flatten-parallel}). Therefore:

\begin{itemize}
    \item \texttt{v.C[i]} places its keys in the \texttt{subtree\_keys} array starting from \texttt{left + positions[2 $\cdot$ i]} (since \texttt{v.C[i]} is \texttt{(2 $\cdot$ i)}-th key source). Thus
    we execute \texttt{Flatten(v.C[i], subtree\_keys, left + positions[2 $\cdot$ i], left + positions[2 $\cdot$ i] + v.C[i].Size)} to let \texttt{C[i]} place its keys in the \texttt{subtree\_keys} array;

    \item If \texttt{v.Exists[i] = false} then \texttt{v.Rep[i]} should not be put in the \texttt{subtree\_keys} array since \texttt{v.Rep[i]} is logically removed;

    \item Otherwise, \texttt{v.Exists[i] = true} and \texttt{v.Rep[i]} should be placed at \texttt{subtree\_keys[left + positions[2 $\cdot$ i + 1]]} since \texttt{v.Rep[i]} is the \texttt{(2 $\cdot$ i + 1)}-th key source.
\end{itemize}

Each key source can be processed in parallel, since there are no data dependencies between them.

\begin{figure}
  \centering
  \caption{Parallel flattening of an IST node} 
  \label{flatten-parallel}
  \includegraphics[width=\linewidth]{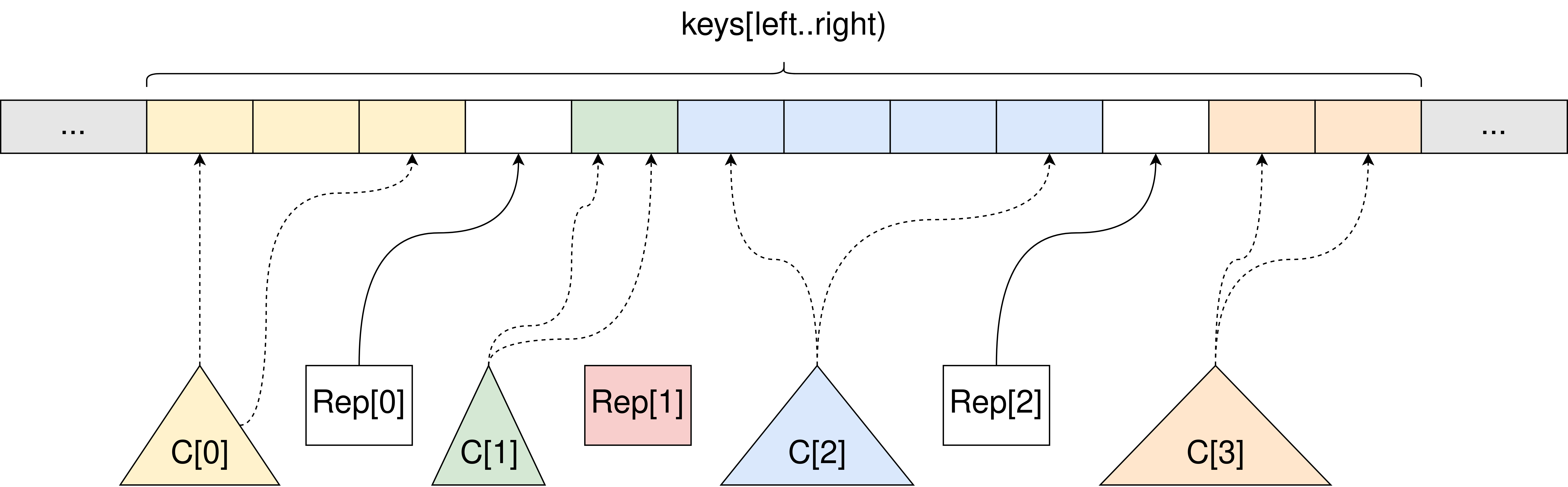}
\end{figure}

In Fig.~\ref{flatten-parallel}, \texttt{C[0]} will place its keys in \texttt{keys[left .. left + 3)} subarray, \texttt{Rep[0]} will be placed in \texttt{keys[left + 3]}, \texttt{C[1]} will place its keys in \texttt{keys[left + 4 .. left + 5)} subarray, \texttt{Rep[1]} will not be placed in \texttt{keys} array since its logically removed, \texttt{C[2]} will place its keys in \texttt{keys[left + 5 .. left + 9)} subarray, \texttt{Rep[2]} will be placed in \texttt{keys[left + 9]} and \texttt{C[3]} will place its keys in \texttt{keys[left + 10 .. left + 12)} subarray.

\subsection{Building an ideal IST in parallel}
\label{build-section}

Suppose we have a sorted array of keys and we want to build an \emph{ideally balanced IST} (see Section~\ref{insert-balancing-sequenctial-section}) from these keys. We implement this procedure recursively via \texttt{build\_IST\_subarray(keys[], left, right)} procedure~--- it builds an ideal IST containing keys from the \texttt{keys[left..right)} subarray and returns the root of the newly-built subtree. Thus, to build a new subtree from array \texttt{keys} we just use \texttt{new\_root := build\_IST\_subarray(keys[], 0, |keys|)} to obtain the root of the IST built from the whole \texttt{keys} array.

If the size of the subarray (i.e., \texttt{right - left}) is less than some constant $H$, we return a leaf node containing all the keys from \texttt{keys[left..right)} in its \texttt{Rep} array. 

Otherwise (i.e., if \texttt{right - left} is big enough), we have to build a non-leaf node. Let us denote \texttt{m := right - left; k := $\lfloor \sqrt{m} \rfloor - 1$}. As follows from Definition~\ref{ideal-definition}, \texttt{Rep} array should have size $\Theta(\sqrt{m})$ and its elements must be equally spaced keys of the initial array. Thus, we copy $k$-th key into \texttt{Rep[0]}, $(2 \cdot k)$-th key into \texttt{Rep[1]}, $(3 \cdot k)$-th key into \texttt{Rep[2]} and so on~--- in general, we copy $(i + 1) \cdot k$-th key into \texttt{Rep[i]}.
Note, our subarray \texttt{keys[left..right)} begins at position \texttt{left} of the initial array \texttt{keys}. Thus, we copy key \texttt{keys[left + (i + 1) $\cdot$ k]} into \texttt{Rep[i]}. All the copying can be done in parallel since there are no data dependencies. This way we obtain \texttt{Rep} array of size $\Theta(\sqrt{m})$ filled with equally-spaced keys of the initial subarray (Fig.~\ref{build-children-pic}).


Now we should build the children of the newly-created node (Fig.~\ref{build-children-pic}):

\begin{itemize}
    \item \texttt{Rep[0] = keys[left + k]}. Thus, all keys less than \texttt{keys[left + k]} will be stored in \texttt{C[0]} subtree (see Section~\ref{IST-section} for details). Since \texttt{keys} array is sorted, \texttt{C[0]} must be built from \texttt{keys[left..left + k)}, thus \texttt{C[0] = build\_IST\_subarray(keys[], left, left + k)};

    \item for $1 \leq i \leq k - 2$, \texttt{Rep[i - 1] = keys[left + i $\cdot$ k]} and \texttt{Rep[i] = keys[left + (i + 1) $\cdot$ k]}. Thus, all keys \texttt{x} such that \texttt{Rep[i-1] < x < Rep[i]} should be stored in \texttt{C[i]} subtree. Since \texttt{keys} array is sorted, \texttt{C[i]} must be built from the subarray \texttt{keys[left + i $\cdot$ k + 1..left + (i + 1) $\cdot$ k)}, thus \texttt{C[i] = build\_IST\_subarray(keys[], left + i $\cdot$ k + 1, left + (i + 1) $\cdot$ k)};

    \item \texttt{Rep[k - 1] = keys[left + $k^2$]} (note, that \texttt{left + $k^2$ = left + ${\left(\lfloor \sqrt{\texttt{right} - \texttt{left}} \rfloor - 1 \right)}^2 < \texttt{right}$}).
    Thus, all keys greater than \texttt{keys[left + $k^2$]} are stored in \texttt{C[k]} subtree. Since \texttt{keys} array is sorted, \texttt{C[k]} must be built from \texttt{keys[left + $k^2$..right)}, thus \texttt{C[k] = build\_IST\_subarray(keys[], left + $k^2$, right)}.
\end{itemize}

We can build all children in parallel, since there is no data dependencies between them.

\begin{figure}[H]
  \centering
  \caption{Building children of a new node} 
  \label{build-children-pic}
  \includegraphics[width=\linewidth]{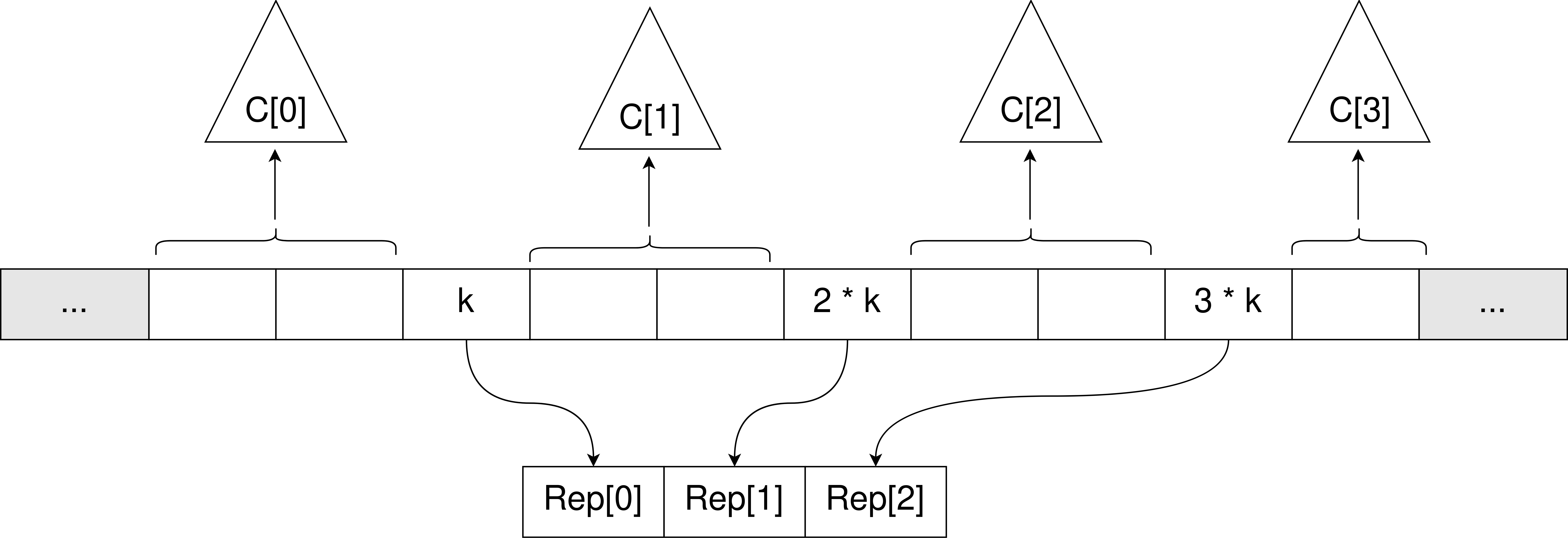}
\end{figure}

To finish the construction of a node we need to calculate \texttt{node.ID} array described in Section~\ref{interpolation-search-section}. We can build it in the following way:

\begin{enumerate}
    \item Create an array \texttt{bounds} of size \texttt{m + 1} such that \texttt{bound[i]} = $\texttt{keys[left]} + i \cdot (\texttt{keys[right - 1]} - \texttt{keys[left]}) / m$ where $m = \Theta(n^\varepsilon)$ with some $\varepsilon \in [\frac{1}{2}; 1)$;

    \item Use \texttt{ID := Rank(bounds, node.Rep)}. After that \texttt{ID[i] = j} iff $Rep[j] \leq \texttt{bounds[i]} < Rep[j + 1]$
\end{enumerate}

\section{Theoretical results}
\label{theoretical-section}
In this section, we present the theoretical bounds for our data structure. These bounds are quite trivial, so, we just explain the intuition behind them.

\begin{theorem}
The flatten operation of an IST with $n$ keys has $O(n)$ work and $O(\log^3 n)$ span. The building procedure of an ideal IST from an array of size $n$ has $O(n)$ work and $O(\log n \cdot \log \log n)$ span. Thus, the rebuilding of IST with $n$ keys costs $O(n)$ work and $O(\log^3 n)$ span.
\end{theorem}
\begin{proof}[Sketch]
While the work bounds are trivial, we are more interested in span bounds.
From~\cite{mehlhorn1993dynamic} we know that in the worst case, the height of IST with $n$ keys does not exceed $O(\log^2 n)$.
Thus, the flatten operation just goes recursively into $O(\log^2 n)$ levels and spends $O(\log n)$ span at each level. This gives $O(\log^3 n)$ span in total. The construction of an ideal IST has $O(\log \log n)$ recursive levels while each level can be executed in $O(\log n)$ time, i.e., copy the elements into \texttt{Rep} array. This gives us the result.
\end{proof}

This brings us closer to our main complexity theorem.
\begin{theorem}
The work of a batched operation on our parallel-batched IST has the same complexity as if we apply all $m$ operations from this batch sequentially to the original IST of size $n$ (from~\cite{mehlhorn1993dynamic}, the expected execution time is $O(m \log \log n)$). The total span of a batched operation is $O(\log^4 n)$.
\end{theorem}
\begin{proof}[Sketch]
The work bound is trivial~--- the only difference with the original IST is that we can rebuild the subtree in advance before applying some of the operations.
Now, we get to the span bounds. From~\cite{mehlhorn1993dynamic}, we know that the height of IST with $n$ keys does not exceed $O(\log^2 n)$. On each level, we spend: 1)~at most $O(\log^2 n)$ span for merge and rank operations; or 2)~we rebuild a subtree at that level and stop. The first part gives us $O(\log^4 n)$ span, while rebuilding takes just $O(\log^3 n)$ span. This leads us to the result of the total $O(\log^4 n)$ span.
\end{proof}

\section{Experiments}
\label{experiments-section}

We have implemented the Parallel Batched IST in C++ using OpenCilk~\cite{blumofe1996cilk} as a framework for fork-join parallelism. 

We tested our parallel-batched IST on three workloads. We initialize the tree with elements from the range $[-10^8; 10^8]$, each taken with probability $1/2$. Thus, the expected size of the tree is $10^8$. Then we execute:

\begin{itemize}
    \item Search for a batch of $10^7$ keys, taken uniformly at random from the range;
    \item Insert a batch of random $10^7$ keys, taken uniformly at random from the range;
    \item Remove a batch of random $10^7$ keys, taken uniformly at random from the range.
\end{itemize}

The experimental results are shown in Fig.~\ref{bench-pic}. The OX axis corresponds to the number of worker processes and the OY axis corresponds to the time required to execute the operation in milliseconds. Each point of the plot is obtained as an average of $10$ runs. We run our code on an Intel Xeon Gold 6230 machine with $16$ cores.

\begin{figure}
  \centering
  \caption{Benchmark results for Parallel-batched Interpolation Search Tree}
  \label{bench-pic}
  \includegraphics[width=0.5\linewidth]{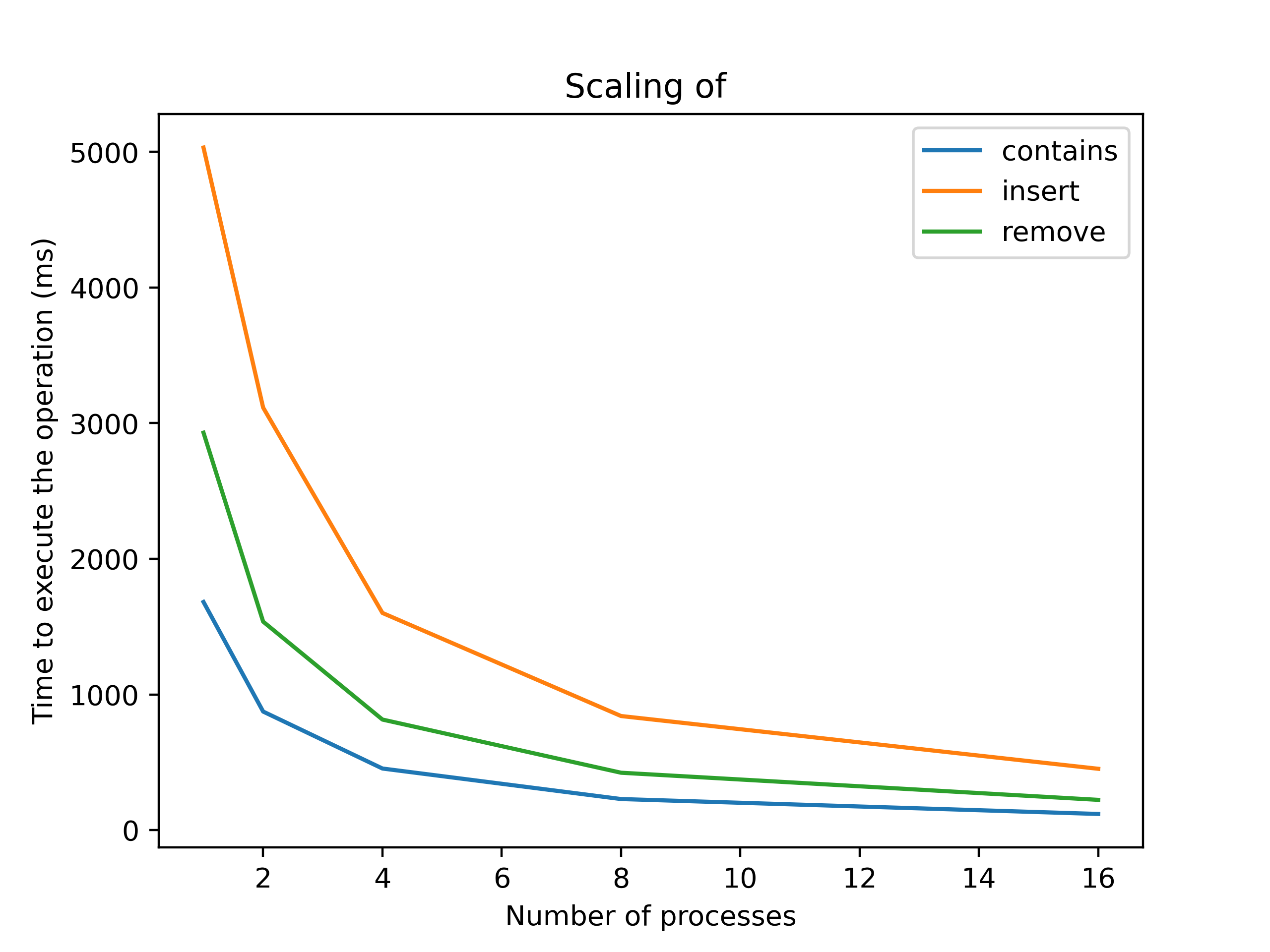}
\end{figure}

As shown in the results, we achieve good scalability. Indeed:

\begin{itemize}
    \item 14x scaling on \texttt{ContainsBatched} operation for 16 processes;
    \item 11x scaling on \texttt{InsertBatched} operation for 16 processes;
    \item 13x scaling on \texttt{RemoveBatched} operation for 16 processes.
\end{itemize}

We also compared our implementation in a sequential mode with \texttt{std::set} implemented via red-back tree. On our machine \texttt{std::set} took 9257 ms to check the existence of  $10^7$ keys in a set with $10^8$ elements while our IST implementation took only 3561 ms using one process. We achieve such speedup by using interpolation search as described in Section~\ref{interpolation-search-section} and by processing high levels of the tree only once for the whole keys batch.

\section{Conclusion}
\label{conclusion-section}

In this work, we presented the first parallel-batched  version of the interpolation search tree that has an optimal work in comparison to the sequential implementation and has a polylogarithmic span. We implemented it and got very promising results. We believe that this work will encourage others to look into parallel-batched data structures based on something more complex than binary search trees.

\bibliographystyle{splncs04}
\bibliography{references.bib}

\end{document}